\documentclass[twoside,leqno,twocolumn]{article}
\pdfoutput=1
\usepackage[letterpaper]{geometry}

\usepackage{ltexpprt}

\usepackage[%
  activate={true,nocompatibility},%
  final,%
  tracking=true,%
  kerning=true,%
  spacing=true,%
  stretch=10,%
  shrink=30]{microtype}
\microtypecontext{spacing=nonfrench}

\usepackage[ruled]{algorithm}

\newcommand{\Bin}{\mathrm{Bin}}
\newcommand{\Oh}{\mathcal{O}}

\usepackage{xcolor}
\usepackage{xspace}
\usepackage{booktabs}
\usepackage{dsfont}
\usepackage{footmisc}
\usepackage{marvosym}
\usepackage{placeins}
\usepackage{rotating}
\usepackage{listings}
\usepackage{amsmath,amssymb}
\let\Pr\relax
\DeclareMathOperator*{\Pr}{\mathds{P}}
\usepackage[switch]{lineno} %
\usepackage{hyperref}
\usepackage[nocompress]{cite}
\usepackage[capitalise,noabbrev]{cleveref}
\usepackage{pgfplots}
\pgfplotsset{compat=newest}

\usepgfplotslibrary{groupplots}
\pgfplotsset{every axis/.style={scale only axis}}

\definecolor{veryLightGrey}{HTML}{F2F2F2}
\definecolor{lightGrey}{HTML}{DDDDDD}
\definecolor{colorSimdRecSplit}{HTML}{444444}
\definecolor{colorChd}{HTML}{377EB8}
\definecolor{colorRustFmph}{HTML}{A65628}
\definecolor{colorRustFmphGo}{HTML}{A65628}
\definecolor{colorSicHash}{HTML}{4DAF4A}
\definecolor{colorPthash}{HTML}{984EA3}
\definecolor{colorRecSplit}{HTML}{FF7F00}
\definecolor{colorBbhash}{HTML}{F781BF}
\definecolor{colorShockHash}{HTML}{F8BA01}

\colorlet{colorBruteForce}{colorSicHash}
\colorlet{colorRotationFitting}{colorChd}

\pgfdeclareimage[interpolate,height=1.85mm,width=1.85mm]{shockhashMark}{fig/shockhash}
\pgfdeclareplotmark{shockhash}{\pgftext[at=\pgfpointorigin]{\pgfuseimage{shockhashMark}}}

\pgfplotscreateplotcyclelist{myColorList}{%
  colorSicHash,mark=o\\%
  colorRecSplit,mark=triangle\\%
  colorSimdRecSplit,mark=pentagon\\%
  colorChd,mark=square\\%
  colorPthash,mark=x\\%
  colorBbhash,mark=diamond\\%
}
\pgfplotscreateplotcyclelist{transparentHeatmap}{%
  colorSicHash,mark=*,fill=colorSicHash\\
}
\pgfdeclareplotmark{flippedTriangle}{%
  \pgfpathmoveto{\pgfqpointpolar{-90}{1.2\pgfplotmarksize}}%
  \pgfpathlineto{\pgfqpointpolar{30}{1.2\pgfplotmarksize}}%
  \pgfpathlineto{\pgfqpointpolar{150}{1.2\pgfplotmarksize}}%
  \pgfpathclose%
  \pgfusepath{stroke}
}

\pgfplotsset{
  mark repeat*/.style={
    scatter,
    scatter src=x,
    scatter/@pre marker code/.code={
      \pgfmathtruncatemacro\usemark{
        or(mod(\coordindex,#1)==0, (\coordindex==(\numcoords-1))
      }
      \ifnum\usemark=0
        \pgfplotsset{mark=none}
      \fi
    },
    scatter/@post marker code/.code={}
  },
  major grid style={thin,dotted},
  minor grid style={thin,dotted},
  ymajorgrids,
  yminorgrids,
  every axis/.append style={
    line width=0.7pt,
    tick style={
      line cap=round,
      thin,
      major tick length=4pt,
      minor tick length=2pt,
    },
    mark options={solid},
  },
  legend cell align=left,
  legend style={
    line width=0.7pt,
    /tikz/every even column/.append style={column sep=3mm,black},
    /tikz/every odd column/.append style={black},
    mark options={solid},
  },
  legend style={font=\small},
  title style={yshift=-2pt},
  enlarge x limits=0.04,
  every tick label/.append style={font=\footnotesize},
  every axis label/.append style={font=\small},
  every axis y label/.append style={yshift=-1ex},
  /pgf/number format/1000 sep={},
  axis lines*=left,
  xlabel near ticks,
  ylabel near ticks,
  axis lines*=left,
  label style={font=\footnotesize},
  tick label style={font=\footnotesize},
  cycle list name=myColorList,
  plotParameters/.style={
    width=35.0mm,
    height=30.0mm,
  },
  plotScaling/.style={
    width=38.0mm,
    height=30.0mm,
  },
  plotScalingConfigs/.style={
    width=35.0mm,
    height=30.0mm,
  },
  plotPareto/.style={
    width=65.0mm,
    height=40.0mm,
  },
  plotHfEvals/.style={
    width=28.0mm,
    height=30.0mm,
  },
  plotLeafMethods/.style={
    width=29.0mm,
    height=30.0mm,
  },
  plotProbabilities/.style={
    width=35.0mm,
    height=30.0mm,
    only marks,
    mark size=.75pt,
    cycle list name=transparentHeatmap,
  },
}

\usepackage{newunicodechar}
\usepackage{silence}
\newunicodechar{♢}{\tikz \node[inner sep=1.5,draw,diamond] {};}
\newunicodechar{☆}{\tikz \node[inner sep=1,draw,star,star point ratio=2] {};}
\newunicodechar{△}{\triangle}
\newunicodechar{⬜}{\kern 0.5pt\tikz \node[inner sep=1.7,draw,regular polygon,regular polygon sides=4] {};\kern 0.5pt}
\newunicodechar{◯}{\tikz[baseline=-3pt] \node[inner sep=1.7,draw,cloud,cloud puffs=4,cloud puff arc=190] {};}

\newunicodechar{⊥}{\bot}
\newunicodechar{•}{\item}
\newunicodechar{✓}{\checkmark}
\newunicodechar{✗}{\xmark}
\newunicodechar{…}{\dots}
\newunicodechar{≔}{\coloneqq}
\newunicodechar{⁻}{^-}
\newunicodechar{⁺}{^+}
\newunicodechar{₋}{_-}
\newunicodechar{₊}{_+}
\newunicodechar{ℓ}{\ell}
\newunicodechar{•}{\item}
\newunicodechar{…}{\dots}
\newunicodechar{≔}{\coloneqq}
\newif\ifnormopen\normopenfalse
\newunicodechar{‖}{\ifnormopen\rVert\normopenfalse\else\rVert\normopentrue\fi}
\newunicodechar{≤}{\leq}
\newunicodechar{≥}{\geq}
\newunicodechar{≰}{\nleq}
\newunicodechar{≱}{\ngeq}
\newunicodechar{⊕}{\oplus}
\newunicodechar{⊗}{\otimes}
\newunicodechar{≠}{\neq}
\newunicodechar{¬}{\neg}
\newunicodechar{≡}{\equiv}
\newunicodechar{₀}{_0}
\newunicodechar{₁}{_1}
\newunicodechar{₂}{_2}
\newunicodechar{₃}{_3}
\newunicodechar{₄}{_4}
\newunicodechar{₅}{_5}
\newunicodechar{₆}{_6}
\newunicodechar{₇}{_7}
\newunicodechar{₈}{_8}
\newunicodechar{₉}{_9}
\newunicodechar{ₚ}{_p}
\newunicodechar{ₙ}{_n}
\newunicodechar{ₐ}{_a}
\newunicodechar{ₑ}{_e}
\newunicodechar{ₕ}{_h}
\newunicodechar{ₖ}{_k}
\newunicodechar{ₗ}{_l}
\newunicodechar{ₘ}{_m}
\newunicodechar{ₛ}{_s}
\newunicodechar{ₜ}{_t}
\newunicodechar{ₓ}{_x}
\newunicodechar{⁰}{^0}
\newunicodechar{¹}{^1}
\newunicodechar{²}{^2}
\newunicodechar{³}{^3}
\newunicodechar{⁴}{^4}
\newunicodechar{⁵}{^5}
\newunicodechar{⁶}{^6}
\newunicodechar{⁷}{^7}
\newunicodechar{⁸}{^8}
\newunicodechar{⁹}{^9}
\newunicodechar{∈}{\in}
\newunicodechar{∉}{\notin}
\newunicodechar{⊂}{\subset}
\newunicodechar{⊃}{\supset}
\newunicodechar{⊆}{\subseteq}
\newunicodechar{⊇}{\supseteq}
\newunicodechar{⊄}{\nsubset}
\newunicodechar{⊅}{\nsupset}
\newunicodechar{⊈}{\nsubseteq}
\newunicodechar{⊉}{\nsupseteq}
\newunicodechar{∪}{\cup}
\newunicodechar{∩}{\cap}
\newunicodechar{∀}{\forall}
\newunicodechar{∃}{\exists}
\newunicodechar{∄}{\nexists}
\newunicodechar{∨}{\vee}
\newunicodechar{∧}{\wedge}
\newunicodechar{ℝ}{\mathbb{R}}
\newunicodechar{ℙ}{\mathbb{P}}
\newunicodechar{ℕ}{\mathbb{N}}
\newunicodechar{𝔼}{\mathbb{E}}
\newunicodechar{𝔽}{\mathbb{F}}
\newunicodechar{ℤ}{\mathbb{Z}}
\newunicodechar{⌊}{\lfloor}
\newunicodechar{⌋}{\rfloor}
\newunicodechar{⌈}{\lceil}
\newunicodechar{⌉}{\rceil}
\newunicodechar{·}{\cdot}
\newunicodechar{∘}{\circ}
\newunicodechar{×}{\times}
\newunicodechar{↑}{\uparrow}
\newunicodechar{↓}{\downarrow}
\newunicodechar{→}{\rightarrow}
\newunicodechar{←}{\leftarrow}
\newunicodechar{⇒}{\Rightarrow}
\newunicodechar{⇐}{\Leftarrow}
\newunicodechar{↔}{\leftrightarrow}
\newunicodechar{⇔}{\Leftrightarrow}
\newunicodechar{↦}{\mapsto}
\newunicodechar{∅}{\varnothing}
\newunicodechar{∞}{\infty}
\newunicodechar{≅}{\cong}
\newunicodechar{≈}{\approx}
\newunicodechar{ℓ}{\ell}
\newunicodechar{𝟙}{\mathds{1}}
\newunicodechar{𝟘}{\mathds{0}}
\newunicodechar{↪}{\hookrightarrow}

\newunicodechar{α}{\alpha}
\newunicodechar{β}{\beta}
\newunicodechar{γ}{\gamma}
\newunicodechar{Γ}{\Gamma}
\newunicodechar{δ}{\delta}
\newunicodechar{Δ}{\Delta}
\newunicodechar{ε}{\varepsilon}
\newunicodechar{ζ}{\zeta}
\newunicodechar{η}{\eta}
\newunicodechar{θ}{\theta}
\newunicodechar{Θ}{\Theta}
\newunicodechar{ι}{\iota}
\newunicodechar{κ}{\kappa}
\newunicodechar{λ}{\lambda}
\newunicodechar{Λ}{\Lambda}
\newunicodechar{μ}{\mu}
\newunicodechar{ν}{\nu}
\newunicodechar{ξ}{\xi}
\newunicodechar{Ξ}{\Xi}
\newunicodechar{π}{\pi}
\newunicodechar{Π}{\Pi}
\newunicodechar{ρ}{\rho}
\newunicodechar{σ}{\sigma}
\newunicodechar{Σ}{\Sigma}
\newunicodechar{τ}{\tau}
\newunicodechar{υ}{\upsilon}
\newunicodechar{ϒ}{\Upsilon}
\newunicodechar{φ}{\varphi}
\newunicodechar{ϕ}{\phi}
\newunicodechar{Φ}{\Phi}
\newunicodechar{χ}{\chi}
\newunicodechar{ψ}{\psi}
\newunicodechar{Ψ}{\Psi}
\newunicodechar{ω}{\omega}
\newunicodechar{Ω}{\Omega}

\newunicodechar{𝒜}{\mathcal{A}}
\newunicodechar{ℬ}{\mathcal{B}}
\newunicodechar{𝒞}{\mathcal{C}}
\newunicodechar{𝒟}{\mathcal{D}}
\newunicodechar{ℰ}{\mathcal{E}}
\newunicodechar{ℱ}{\mathcal{F}}
\newunicodechar{𝒢}{\mathcal{G}}
\newunicodechar{ℋ}{\mathcal{H}}
\newunicodechar{ℐ}{\mathcal{I}}
\newunicodechar{𝒥}{\mathcal{J}}
\newunicodechar{𝒦}{\mathcal{K}}
\newunicodechar{ℒ}{\mathcal{L}}
\newunicodechar{ℳ}{\mathcal{M}}
\newunicodechar{𝒩}{\mathcal{N}}
\newunicodechar{𝒪}{\mathcal{O}}
\newunicodechar{𝒫}{\mathcal{P}}
\newunicodechar{𝒬}{\mathcal{Q}}
\newunicodechar{ℛ}{\mathcal{R}}
\newunicodechar{𝒮}{\mathcal{S}}
\newunicodechar{𝒯}{\mathcal{T}}
\newunicodechar{𝒰}{\mathcal{U}}
\newunicodechar{𝒱}{\mathcal{V}}
\newunicodechar{𝒲}{\mathcal{W}}
\newunicodechar{𝒳}{\mathcal{X}}
\newunicodechar{𝒴}{\mathcal{Y}}
\newunicodechar{𝒵}{\mathcal{Z}}
\newunicodechar{𝒶}{\mathcal{a}}
\newunicodechar{𝒷}{\mathcal{b}}
\newunicodechar{𝒸}{\mathcal{c}}
\newunicodechar{𝒹}{\mathcal{d}}
\newunicodechar{ℯ}{\mathcal{e}}
\newunicodechar{𝒻}{\mathcal{f}}
\newunicodechar{ℊ}{\mathcal{g}}
\newunicodechar{𝒽}{\mathcal{h}}
\newunicodechar{𝒾}{\mathcal{i}}
\newunicodechar{𝒿}{\mathcal{j}}
\newunicodechar{𝓀}{\mathcal{k}}
\newunicodechar{𝓁}{\mathcal{l}}
\newunicodechar{𝓂}{\mathcal{m}}
\newunicodechar{𝓃}{\mathcal{n}}
\newunicodechar{ℴ}{\mathcal{o}}
\newunicodechar{𝓅}{\mathcal{p}}
\newunicodechar{𝓆}{\mathcal{q}}
\newunicodechar{𝓇}{\mathcal{r}}
\newunicodechar{𝓈}{\mathcal{s}}
\newunicodechar{𝓉}{\mathcal{t}}
\newunicodechar{𝓊}{\mathcal{u}}
\newunicodechar{𝓋}{\mathcal{v}}
\newunicodechar{𝓌}{\mathcal{w}}
\newunicodechar{𝓍}{\mathcal{x}}
\newunicodechar{𝓎}{\mathcal{y}}
\newunicodechar{𝓏}{\mathcal{z}}

\crefname{listing}{Algorithm}{Algorithms}
\crefname{@theorem}{Theorem}{Theorems}
\crefname{@figure}{Figure}{Figures}
\crefname{@table}{Table}{Tables}

\newcommand{\myparagraph}[1]{\subparagraph*{#1}}

\let\oldcite\cite
\renewcommand\cite{\unskip~\oldcite}

\newcommand{\mytitle}{ShockHash:\\ Towards Optimal-Space Minimal Perfect Hashing Beyond Brute-Force}
\title{\mytitle}

\newcommand{\mythanks}[3]{\thanks{#1 \href{mailto:#2}{\includegraphics[height=8pt]{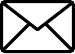}} \href{https://orcid.org/#3}{\includegraphics[height=8pt]{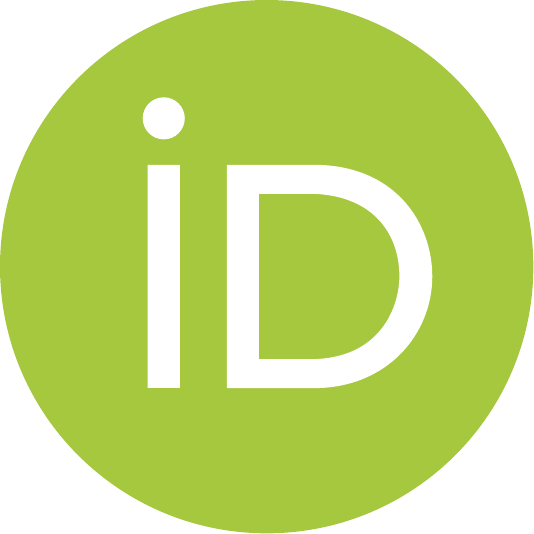}}}}
\date{}
\author{
  Hans-Peter Lehmann\mythanks{Karlsruhe Institute of Technology, Germany.}{hans-peter.lehmann@kit.edu}{0000-0002-0474-1805}
  \and
  Peter Sanders\mythanks{Karlsruhe Institute of Technology, Germany.}{sanders@kit.edu}{0000-0003-3330-9349}
  \and
  Stefan Walzer\mythanks{Karlsruhe Institute of Technology, Germany.}{stefan.walzer@kit.edu}{0000-0002-6477-0106}
}

\hypersetup{
  colorlinks=true,
  pdftitle={\mytitle},
  pdfauthor={Lehmann, Sanders, Walzer},
  pdfsubject={}
}

\begin{document}
\maketitle 
\def\speedupSIMD156{75}

\def\maxSpeedupNonSimdPlainRotate{195}

\def\maxSpeedupNonSimdPlain{25}

\begin{abstract}
  A minimal perfect hash function (MPHF) maps a
  set~$S$ of $n$ keys to the first $n$ integers
  without collisions.  There is a lower bound of
  $n\log_2e-\Oh(\log n)$ bits of space needed to
  represent an MPHF.  A matching upper bound
  is obtained using the \emph{brute-force} algorithm that
  tries random hash functions until stumbling on an MPHF and
  stores that function's seed.
  In expectation, $e^n\textrm{poly}(n)$
  seeds need to be tested.
  The most space-efficient previous
  algorithms for constructing MPHFs all use such a
  brute-force approach as a basic building
  block.

  In this paper, we introduce ShockHash -- {\bf
    S}mall, {\bf h}eavily {\bf o}verloaded cu{\bf
    ck}oo {\bf Hash} tables. ShockHash uses two hash functions $h₀$ and $h₁$, hoping for the existence of a function $f : S → \{0,1\}$ such that $x ↦ h_{f(x)}(x)$ is an MPHF on $S$.
  In graph terminology, ShockHash generates $n$-edge random
  graphs until stumbling on a \emph{pseudoforest}
  -- a graph where each component contains as many
  edges as nodes. Using cuckoo hashing,
  ShockHash then derives an MPHF from the
  pseudoforest in linear time.  It uses a 1-bit
  retrieval data structure to store $f$ using $n + o(n)$~bits.

  By carefully
  analyzing the probability that a random graph
  is a pseudoforest, we show that
  ShockHash needs to try only
  $(e/2)^n\textrm{poly}(n)$
  hash function seeds in expectation.
  This reduces the space for storing the seed by roughly $n$ bits (maintaining the asymptotically optimal space consumption) and speeds up construction by almost a
  factor of $2^n$ compared to brute-force.
  When using ShockHash as a building block within
  the RecSplit framework we obtain ShockHash-RS, which yields the currently most
  space efficient MPHFs, i.e.,
  competing approaches need about two orders of magnitude more
  work to achieve the same~space.
\end{abstract}

\section{Introduction}
A perfect hash function (PHF) maps a set of $N$ keys to the first $M$ integers without collisions.
If $M=N$, the hash function is called \emph{minimal} perfect (MPHF) and is a bijection between the keys and the $N$ first integers $[N]$.

Minimal perfect hashing has many applications.
For example, it can be used to implement \emph{hash tables} with guaranteed constant access time \cite{fredman1984storing}.
Storing only payload data in the hash table cells, we obtain an \emph{updatable retrieval data structure} \cite{MSSZ14}, and storing only fingerprints \cite{fan2014cuckoo,bender2018bloom}, we obtain an \emph{approximate membership data structure}.
Finally, the hashes can be used as small \emph{identifiers} of the input keys \cite{botelho2007perfect}, which are more efficient to deal with than large and complex keys.

There is a lower bound of about $N\log_2 e\approx
1.44N$ bits
needed to represent an MPHF. There is also a
theoretical construction matching this bound that
runs in linear time and allows constant query time
\cite{HagTho01}.  However, this construction does
not work for realistic $N$, so that (minimal)
perfect hashing remains an interesting topic for
algorithm engineering. A long sequence of previous
work has developed a range of practical approaches
with different space-time tradeoffs.

Many approaches first construct an outer hash function $g : S → \{1,…,k\}$ that partitions the input set $S$ into small subsets $S₁, S₂, …, S_k$ of sizes $s₁ ≈ s₂ ≈ … ≈ s_k$ and then constructs a perfect hash function $h_i : S_i → \{1,…, s_i\}$ on each subset. Given $s_i$, or better yet the prefix sums $p_i = s₀ + \cdots s_{i-1}$, an MPHF on $S$ is given as $x ↦ p_{g(x)} + h_{g(x)}(x)$.

On the one hand, there are holistic methods where such a partitioning step is not essential (though still possibly useful). These are (so far) all
a constant factor away from the space lower bound
(e.g. \cite{CSLSR11,MSSZ14,lehmann2023sichash,beling2023fingerprinting}).  One of the most space-efficient
approaches among these is \emph{SicHash} \cite{lehmann2023sichash}
that
maps $N$ keys to $N(1+ε)$ unique table
entries using a generalization of cuckoo hashing
\cite{pagh2004cuckoo,fotakis2005space}.  The
choice of hash function for each key is then stored in a
retrieval data structure and a ranking data structure turns the constructed PHF into an MPHF. SicHash is not
space-optimal partly because the cuckoo table tends to admit many valid placements of its keys, meaning a single input set is redundantly handled by many distinct states of the PHF data structure.

On the other hand, there are methods that use
brute-force trial-and-error of hash functions on
subsets of size $n$ \cite{esposito2020recsplit,bez2022high,fox1992faster,belazzougui2009hash,pibiri2021pthash}, which takes roughly $e^n$ trials.
Hence, an aggressive partitioning step is required to obtain an acceptable overall running time.
The currently most space efficient approach,
\emph{RecSplit} \cite{esposito2020recsplit}, is of this kind and
recursively splits the input set into very small ($n\approx 16$)
\emph{leaf} subsets.
Surprisingly, when
using aggressive parallelization, this enables
higher construction throughput than the best
holistic methods even when being fairly far away
from the space lower bound \cite{bez2022high}.

In this paper, we introduce \emph{ShockHash} --
{\bf S}mall, {\bf h}eavily {\bf o}verloaded cu{\bf
  ck}oo {\bf Hash} tables, which can be seen as an extreme
  version of SicHash where we use two hash functions for each key
and retry construction until we can completely fill the cuckoo hash table.
That way, we achieve an MPHF without an intermediate non-minimal PHF.
In graph terminology, ShockHash repeatedly generates an $n$-edge random graph where each key corresponds to one edge, connecting the candidate positions of the key.
The table can be filled if and only if the graph is a \emph{pseudoforest}
  -- a graph where each component contains as many
  edges as nodes.
The ShockHash idea is
straightforward in principle but interesting in
two respects. First, we can \emph{prove} that when
using basic cuckoo hashing with two choices (and
thus 1-bit retrieval)
there is only an insignificant amount of
redundancy.
Therefore, ShockHash approaches the
information theoretic lower bound for large
$n$ and has running time $(e/2)^n·\textrm{poly}(n)$ (nearly a factor $2^n$ faster than brute-force).

We obtain ShockHash-RS by using ShockHash instead of brute-force as a base case within the RecSplit framework.
This brings algorithmic sophistication back into business: Though there is a small penalty in query time due to the additional access to a retrieval data structure, ShockHash-RS is now the most efficient method to construct nearly space-optimal MPHFs.
At the most space efficient configurations (leaf size $n \approx 40$), ShockHash-RS construction is about two orders of magnitude faster than tuned RecSplit \cite{bez2022high} when achieving the same space efficiency and using the same architecture.
An important step in this
harmonization of theory and practice is the
observation that only an exponentially small
fraction of the hash functions tried by ShockHash require the construction of a cuckoo hash table. The other
cases can be covered with a simple bit-parallel
filter that checks whether all entries of the
cuckoo table are hit by some key.  This
removes much of the time overhead which made brute-force seemingly superior.

In \cref{s:prelim}, we introduce basic concepts like cuckoo hashing and retrieval data structures.
In \cref{s:related}, we discuss related MPHF constructions from the literature.
Then we introduce ShockHash in \cref{s:shockhash}, and prove its space usage and construction time in \cref{s:analysis}.
\Cref{s:shockhashInRecSplit} explains ShockHash-RS, i.e.\ how we integrate ShockHash into the RecSplit framework. \Cref{s:refinements} gives further refinements including an outline on how to parallelize construction and how to integrate ShockHash into other perfect hashing approaches.
Finally, we compare an implementation of ShockHash-RS with competitors from the literature in \cref{s:experiments}.
We give a conclusion and an outlook to future work in \cref{s:conclusion}.

\myparagraph{Summary of Contributions}
\begin{itemize}\itemsep-0.8mm
\vspace{-2mm}
\item Building block for space-efficient construction of MPHFs as an alternative to brute-force, with a speedup of $2^n$
\item Theoretical analysis explaining the experiments
\item Efficient integration into the RecSplit framework
\item Careful experimental comparison with the
  state-of-the-art. Competitors need two orders of magnitude
  more work to achieve similar space efficiency
\item ShockHash can be integrated into further frameworks
\end{itemize}

\section{Preliminaries}\label{s:prelim}
In the following section, we explain basic ingredients of ShockHash.
This also includes the two perfect hash function constructions SicHash \cite{lehmann2023sichash} and RecSplit \cite{esposito2020recsplit} that ShockHash-RS is based on.

\myparagraph{Cuckoo Hashing.}  Cuckoo Hashing
\cite{pagh2004cuckoo} is a well known approach to
handle collisions in hash tables.  Each object
gets two candidate cells via two hash functions.
A query operation looks at the two cells
and searches for the object.  If an insertion
operation tries to insert an object into a cell
that is already full, the object already stored in
the cell is taken out and recursively inserted
using its other candidate position.  Cuckoo
hashing can be extended to use more than two hash
functions \cite{fotakis2005space}, or cells with
more than one object in them
\cite{dietzfelbinger2007balanced}.  In this paper,
we are only interested in the basic version with
two hash functions and one object per cell.

\label{s:loadThresholds}
The load threshold  of a cuckoo hash table \cite{L:A_New_Approach:2012,FKP:The_Multiple:2011,fountoulakis2012sharp} is the number of keys that can be inserted before insertion likely fails.
For cuckoo hashing with two candidate cells, the load threshold is $c=0.5$.

\myparagraph{Pseudoforests.}
\label{s:pseudotrees}
Cuckoo hashing can be modeled as a random graph $G$, where each node represents a table cell and each edge represents one object, connecting its two candidate cells.
It is easy to see that a cuckoo hash table can be constructed successfully if and only if the edges of $G$ can be oriented such that the indegree of each node is $\leq 1$.
In the following, we call this a 1-orientation.
A 1-orientation exists if and only if $G$ is a \emph{pseudoforest}, i.e.\ every connected component of $G$ is a pseudotree. A pseudotree is either a tree or a cycle with trees branching from it.
A way to check whether a graph is a pseudoforest is to check whether each component contains at most as many edges as nodes.

\myparagraph{Retrieval Data
  Structures.}\label{s:retrieval} For a given set
$S$ of $N$ keys, a retrieval data structure (or
\emph{static function} data structure) stores a function
$S\rightarrow\{0, 1\}^r$ that maps each key to a
specific $r$-bit value.  Because it may %
return arbitrary values for keys not in $S$, it is
possible to represent the function without
representing $S$ itself.  Representing a retrieval
data structure needs at least $rN$ bits of space
and there are practical data structures that need
$rN+o(rN)$ bits allowing linear construction time
and constant query time.  In particular, for $r = 1$,
\emph{Bumped Ribbon Retrieval (BuRR)} \cite{dillinger2022burr}
reduces function evaluation to \texttt{XOR}ing a
hash function value with a segment of a
precomputed table and reporting the parity of the
result. This table can be determined by solving a
nearly diagonal system of linear equations (a
``ribbon''). In practice, BuRR has a space
overhead of about $1$\%.

\myparagraph{SicHash.}
{\bf S}mall {\bf
  i}rregular {\bf c}uckoo tables for perfect {\bf
  Hash}ing \cite{lehmann2023sichash} constructs
perfect hash functions through cuckoo hashing.
It constructs a cuckoo hash
table and then uses a retrieval data structure to
store which of the hash function choices was
finally used for each key.
SicHash's main innovation is using a
careful mix of \mbox{1--3} bit retrieval data
structures. It achieves a favorable
space-performance tradeoff when being allowed 2--3
bits of space per key.  It cannot go below
this because using only 1-bit retrieval seems to
lead far from minimality while using 2 or more
bits for retrieval allows redundant choices that
cannot achieve space-optimality.  SicHash achieves
a rather limited gain in space efficiency by
\emph{overloading} the table
beyond the load thresholds %
and trying multiple hash functions.
This mainly exploits the variance in the number of keys that can fit.
SicHash leaves the success probability of constructing overloaded tables as an open question.
ShockHash drives the idea of overloading to its extreme and gives a formal analysis for this case.

\myparagraph{RecSplit.}\label{s:recsplit}
RecSplit \cite{esposito2020recsplit} is a minimal perfect hash function that is mainly focused on space efficiency.
First, all keys are hashed to buckets of constant expected size~$b$.
A bucket's set of keys is recursively split into different subsets by searching for a \emph{splitting} hash function that partitions the keys into smaller sets.
This is repeated recursively in a tree-like structure until subproblems (leaves) with constant size $n$ are left (in Ref. \cite{esposito2020recsplit}, the leaf size is called $\ell$).
Within the leaves, RecSplit then performs brute-force search for a minimal perfect hash function (also called \emph{bijection}).
The tree structure is based only on the size of the current bucket.
This makes it possible to store only the seed values for the hash functions without storing structural information.
Apart from encoding overheads for the seeds, this makes RecSplit information theoretically optimal within a bucket.
The number of child nodes (\emph{fanout}) in the two lowest levels is selected such that the amount of brute-force work is balanced between splittings and bijections.

There also is a parallel implementation using multi-threading and SIMD instructions or the GPU \cite{bez2022high}.
The paper also proposes a new technique for searching for bijections called \emph{rotation fitting}.
Instead of just applying hash functions on the keys in a leaf directly, rotation fitting splits up the keys into two sets using a 1-bit hash function.
It then hashes each of the two sets individually, forming two words where the bits indicate which hash values are occupied.
Then it tries to find a way to cyclically rotate the second word, such that the empty positions left by the first set are filled by the positions of the second set.
The paper shows that each rotation essentially gives a new chance for a bijection, so it is a way to quickly evaluate additional hash function seeds.

\section{More Related Work}\label{s:related}
In addition to RecSplit and ShockHash, which we describe in the preliminaries, there is a range of other MPHFs.

\myparagraph{Hash-and-Displace.}
(M)PHFs with Hash-and-Displace \cite{fox1992faster,belazzougui2009hash,pibiri2021pthash} allow fast queries and asymptotically optimal space consumption.
Each key $x$ is first hashed to a small bucket $b(x)$ of keys.
For each bucket $b$, an index $i(b)$
of a hash function $f_{i(b)}$ is stored such $x ↦ f_{i(b(x))}(x)$ is an injective function.
For a particular bucket, this index is searched in a brute-force way.
To accelerate the search, buckets are first sorted by their size.
Further acceleration can be achieved by using heterogeneous expected bucket sizes.
PTHash \cite{pibiri2021pthash} is the currently best implementation of this approach.

\myparagraph{Fingerprinting.}
Perfect hashing through fingerprinting \cite{CSLSR11,MSSZ14} hashes the $N$ keys to $\gamma N$ positions using an ordinary hash function, where $\gamma$ is a tuning parameter.
The most space efficient choice $\gamma=1$ leads to a space consumption of $e$ (not $\log_2 e$) bits per key.
A bit vector of length $\gamma N$ indicates positions to which exactly one key was mapped.
Keys that caused collisions are handled recursively in another layer of the same data structure.
At query time, when a key is the only one mapping to its location, a rank operation on the bit vector gives the MPHF value.
Publicly available implementations include BBHash \cite{limasset2017fast} and the significantly faster FMPH \cite{beling2023fingerprinting}.
\mbox{FMPHGO} \cite{beling2023fingerprinting} combines the idea with a few brute-force tries to select a hash function that causes fewer collisions.

\begin{figure*}
  \centering
  \includegraphics[scale=1.0]{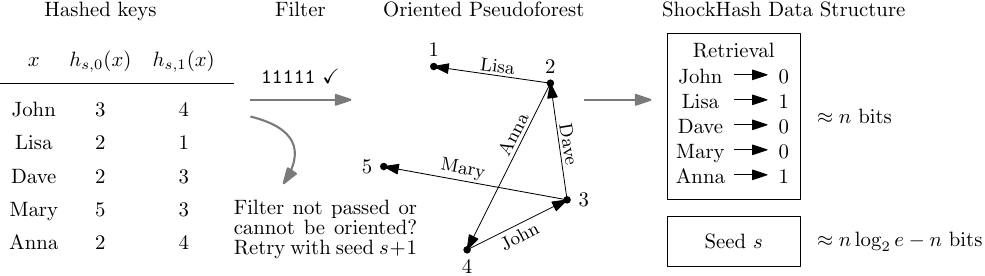}
  \caption{\label{fig:illustration} Illustration of the ShockHash construction. $s$ is a seed value where the resulting graph is a pseudotree. During construction, many seeds need to be tried.
      }
\end{figure*}

\myparagraph{Table Lookup.}
A tempting way to
replace expensive brute-force search is
precomputation of solutions with subsequent table
lookup -- a standard technique used in many
compressed data structures. For a rough idea,
suppose for a subproblem with $n$ keys, we first
map them injectively to a range of size $U' \in \Omega(n^2)$ using an intermediate hash function
(less would lead to collisions -- birthday
paradox). Then, using a lookup table of size
$2^{U'}$, we can find precomputed perfect hash
functions in constant time. However, polynomial running time limits
the subproblem size to $n \in \Oh(\sqrt{\log N})$, where $N$ is the size of the overall input set.
Putting in realistic values, one gets subproblem
size much smaller than what can be easily handled
even with plain RecSplit.  Nevertheless, Hagerup
and Tholey \cite{HagTho01} develop this approach
to a comprehensive theoretical solution of the perfect hashing
problem yielding linear construction time,
constant query time, and space $1+o(1)$ times the
lower bound. However, this method is not even
well-defined for $N<2^{150}$ \cite{botelho2013practical}.
A variant of RecSplit with rotation fitting \cite{bez2022high}
can use lookup tables of size $2^{n}$ to find feasible rotations in constant time.
Unfortunately, this turns out to be slower than trying all rotations directly.

\section{ShockHash}\label{s:shockhash}

We now introduce the main idea of this paper, ShockHash.
ShockHash is briefly mentioned in the extended version of Ref. \cite{bez2022high}, but without any analysis and with an implementation that cannot compete with brute-force.
The asymptotic load threshold of a binary cuckoo hash table is $c=0.5$ (see \cref{s:loadThresholds}), so the success probability of constructing a table with $n$ cells and more than $n/2$ keys tends to zero.
ShockHash overloads a cuckoo hash table far beyond its asymptotic load threshold -- it inserts $n$ keys into a binary cuckoo hash table of size $n$.
As we will see in \cref{lem:constructionTries}, the construction succeeds after $(e/2)^{n} \textrm{poly}(n)$ tries in expectation.
We then record the successful seed
\begin{equation*}\label{eq:ShockHash}
  s=\min\{s ∈ ℕ \mid ∃f ∈ \{0,1\}^S\!: x \!↦\! h_{s,f(x)}(x) \text{ is MPHF}\}
\end{equation*}
and a successful choice $f$ between the two candidate positions of each key.
The seed needs $0.44n+o(n)$ bits in expectation using Golomb-Rice codes \cite{golomb1966run, rice1979some}.
The choices are stored in a $1$-bit retrieval data structure, requiring $n+o(n)$ bits.
This means that the majority of the MPHF description is not stored in the seed, like with the brute-force construction, but in the retrieval data structure.
A query for key $x$ retrieves $f(x)$ from the retrieval data structure and
returns $h_{s,f(x)}(x)$.
\Cref{fig:illustration} gives an illustration of the ShockHash construction.

The beauty of ShockHash is that it
can check $2^n$ different possible hash functions
(determined by the $2^n$ different functions represented by the
retrieval data structure) in linear time.  Refer
to \cref{lem:constructionTries} for details.  This
enables significantly faster construction than brute-force while
still consuming the same amount of space.

\label{s:sccFilter}
As discussed in \cref{s:pseudotrees},
a seed leads to a
successful cuckoo hash table construction if and only if
the corresponding random (multi)graph forms a pseudoforest.  Each component
of size $c$ is a pseudotree if and only if it
contains no more than $c$ edges.  This can be
checked in linear time using connected components
algorithms, or in close to linear time using an
incremental construction of an ordinary cuckoo
hash table.
Nevertheless, the check quickly becomes a bottleneck in practice.

\myparagraph{Filter by Bit Mask.}\label{s:bitmaskFilter}
To address this bottleneck, we therefore use a filter to quickly reject most seeds, namely seeds for which some table cell is not a candidate position of any of the keys.
If there is such a cell, we already know that cuckoo hashing cannot succeed.
Otherwise, cuckoo hashing might succeed.
The filter can be implemented using simple shift and comparison operations.
Also, the filter can use registers, in contrast to the more complex full construction.
It is one of the main ingredients for making ShockHash practical and is easily proven to be very effective:

\begin{lemma}
  \label{lem:filterProbability}
  The probability for a seed to pass the filter,
  i.e.\ for every table cell to be hit by at least
  one key, is at most $(1-e^{-2}+o(1))^n ≈ 0.864^n$.
\end{lemma}
\begin{proof}
  Let $X_i$ denote the number of times that cell
  $i ∈ [n]$ is hit. Then $(X₁,…,Xₙ)$ follows a
  multinomial distribution. The variables
  $X₁,…,Xₙ$ are \emph{negatively associated} in
  the sense introduced in
  \cite{JDP:NegativeAssociation:1983} and satisfy
  \[
    ℙ(∀i∈[n]:X_i ≥ 1) ≤ \prod_{i = 1}^n ℙ(X_i ≥
    1),
  \]
  the intuition being that since the sum $X₁+…+Xₙ
  = 2n$ is fixed, the events $\{X_i ≥ 1\}$ for $i
  ∈ [n]$ are less likely to co-occur compared to
  corresponding independent events.  Since $X_i
  \sim \Bin(2n,\frac 1n)$ for all $i ∈ [n]$ we
  have
  \[
    ℙ(X_i ≥ 1) = 1- (1-\tfrac 1n)^{2n} = 1-e^{-2}
    + o(1) ≈ 0.864
  \]
  and the claim follows.
  \hfill
\end{proof}
A more careful analysis reveals that the probability to pass the filter is around $b^n$ where $b = 2e^λ / (λe^2) ≈ 0.836$ and where $λ ≈ 1.597$ is the soution to $2 = λ/(1-e^{-λ})$. We give a proof in \cref{s:FilterStrength}.

\myparagraph{Rotation Fitting.}\label{s:rotationFitting}
A technique to speed up brute-force search for perfect hash functions is \emph{rotation fitting} \cite{bez2022high} (see \cref{s:recsplit}).
The same idea can be used in ShockHash to accelerate the search.
We distribute the keys to two sets using an ordinary 1-bit hash function.
We then determine the bit mask of output values that are hit in both of the sets.
Like in the bit mask filter, only if the logical \texttt{OR} of both masks has all bits set, it is worth testing the seed more closely.
If we now cyclically rotate one of the bit masks and try again, we basically get a new chance of all output values being hit, without having to hash each key again.
This corresponds to an addition modulo $n$ to all keys of the second set.
We conjecture that -- as in Ref. \cite{bez2022high} -- this reduces the number of hash function evaluations by a factor of $n$, while the space overhead tends to zero.

\section{Analysis}\label{s:analysis}
In this section, we analyze the space usage and construction time of ShockHash.
The main challenge is to lower bound the probability that a hash function seed enables successful construction of the heavily overloaded cuckoo hash table. In the following we assume that a seed is given. We suppress it in notation.

\def\orient{\mathrm{ori}}
\def\PF{\mathrm{PF}}
We are hence given two hash functions $h₀,h₁: S → [n]$ and wish to know the probability that there exists a function $f: S → \{0,1\}$ such that $x ↦ h_{f(x)}(x)$ is bijective, an event we abbreviate with $\orient(f)$. It will be useful to consider the graph
\[
G = ([n], \{\{h₀(x),h₁(x)\} \mid x ∈ S\}).
\]
While similar to an Erdős-Renyi random graph, $G$ may have self-loops\footnote{In our implementation, we avoid self-loops for better performance, but allowing self-loops is easier to analyze.} and multi-edges. Our model matches Model A in \cite{frieze2016introduction}.
There is a one-to-one correspondence between functions $f$ with $\orient(f)$ and $1$-orientations of $G$, i.e.\ ways of directing $G$ such that each node has indegree at most $1$.\footnote{This assumes that there are two ways of directing a self-loop.}

In the following, we give two different proofs for the main result.
\Cref{lem:simpleProof} is a very simple combinatorial argument that we discovered during the review period.
It shows that the probability for $G$ to be 1-orientable is at least $(e/2)^{-n}\sqrt{\pi/(2n)}$.
\Cref{lem:constructionTries} then shows that the probability is at least $(e/2)^{-n} \pi/e$.
Therefore, the simple argument is only a factor of $\Oh(\sqrt{n})$ less tight than the original, much more complex proof.

\begin{theorem}
  \label{lem:simpleProof}
  Let $G$ be a multigraph with $n$ nodes and $n$ edges which are uniformly selected from $[n]^2$.
  Then the probability that $G$ is a pseudotree is at least $(e/2)^{-n}\sqrt{\pi/(2n)}$.
\end{theorem}
\begin{proof}
  The probability space underlying $G$ is that of sampling $2n$ vertices (with replacement) and creating an edge from the samples $2i-1$ and $2i$ for each $i ∈ [n]$.
  For $G$ to be a pseudotree it is sufficient (though not necessary) that the first $n-1$ created edges form a tree.
  There are $n^{n-2}$ labelled $n$-node trees (Cayley's Formula \cite{cayley1878theorem}).
  Since the ordering of the edges and the order of the two samples forming an edge does not matter, each of the trees can be generated in $2^{n-1}(n-1)!$ ways.
  The last two samples can be anything, giving us $n^2$ choices.
  By applying Stirling's approximation, namely
  \begin{align*}
    n! ∈ [\big(\tfrac{n}{e}\big)^{n}\sqrt{2 \pi n}\cdot e^{1/(12n+1)}, \big(\tfrac{n}{e}\big)^{n}\sqrt{2 \pi n} \cdot e^{1/(12n)}],
  \end{align*}
  we can show that the total probability to draw a pseudotree is at least
  \begin{align*}
    \frac{n^{n-2}2^{n-1}(n-1)!n^2}{n^{2n}} \geq \left(\frac{e}{2}\right)^{-n}\sqrt{\pi/(2n)}.
  \end{align*}
  \vspace{-2.0\baselineskip}\\\phantom{foo}\hfill
\end{proof}

For the tighter proof, we write $\PF(G)$ for the event that $G$ is a pseudoforest. As pointed out in \cref{s:pseudotrees}:
\begin{equation}
  \PF(G) ⇔ ∃f: \orient(f).
  \label{eq:pseudoforests-orientable}
\end{equation}
In our case with $n$ nodes and $n$ edges, $\PF(G)$ implies that $G$ is a \emph{maximal} pseudoforest, where every component is a pseudotree and not a tree. Note that a pseudotree that is not a tree admits precisely two $1$-orientations because the unique cycle can be directed in two ways and all other edges must be directed away from the cycle. A useful observation is therefore
\begin{equation}
  \PF(G) ⇒ \# \{f \mid \orient(f) \} = 2^{c(G)}
  \label{eq:orientations-cycles}
\end{equation}
where $c(G)$ is the number of connected components of $G$.

The basic idea of our proof is as follows.
The probability that a random function is minimal perfect is $e^{-n}\textrm{poly}(n)$ (see \cref{lem:bruteForce}).
Each of the $2^n$ functions $f: S → \{0,1\}$ has that chance of satisfying $\orient(f)$ and yielding an MPHF.
However, simply multiplying $e^{-n}\textrm{poly}(n)$ by $2^n$ does not necessarily yield an approximation for the probability that such an $f$ exists.
The key point here is that the $2^n$ functions determined by the $2^n$ different options for $f$ are correlated.
If there are some graphs with many 1-orientations, we may find many MPHFs at once and the probability that at least one $1$-orientation exists is reduced.
A key step will therefore be to show that we tend to find only a few MPHFs at once. This amounts to analyzing the distribution of the number of components in random maximal pseudoforests.
The main proof in \cref{lem:constructionTries} then formally bounds the probability that a random graph can be 1-oriented, juggling different probability spaces.

In the following it will be useful to reveal $G$ in two steps. First the degree of each node is revealed, by randomly distributing $2n$ \emph{stubs} (or half-edges) among the $n$ nodes. This yields a so-called \emph{configuration model} \cite{newman2010networks} from which the edges are then obtained by randomly matching the stubs. The following Lemma should clarify what exactly it is we need.

\begin{lemma}
  \label{lem:probability-spaces}
  Let $x_1,…,x_{2n} ∈ [n]$ be independent and uniformly random. The graphs $G₁,G₂,G₃$ defined in the following have the same distribution as $G$.
  \begin{enumerate}
    • $G₁ = ([n],\{\{x_{2i-1},x_{2i}\} \mid i ∈ [n]\})$.
    • $G₂ = ([n],\{\{x_i,x_j\} | \{i,j\} ∈ M\})$ where $M$ is a uniformly random perfect matching of $[2n]$, i.e.\ a partition of $[2n]$ into $n$ sets of size $2$.
    • $G₃$ is defined like $G₂$, except that $M$ is obtained in a sequence of $n$ rounds. In each round an unmatched number $i ∈ [2n]$ is chosen \emph{arbitrarily} and matched to a distinct unmatched number $j$, chosen uniformly at random. The choice of $i$ may depend on $x₁,…,x_{2n}$ and on the set of numbers matched previously.
  \end{enumerate}
\end{lemma}
The reason for considering these alternative probability spaces for $G$ is that they permit conditioning on partial information about $G$ (such as its degree sequence implicit in $x₁,…,x_{2n}$) but retaining a clean probability space for the remaining randomness.
\begin{proof}
  Compared to $G$, the definition of $G₁$ simply collects the $2n$ relevant hash values in a single list.%
\footnote{Here, we assume that $h₀$ and $h₁$ are fully random hash functions and given for free, which is common in previous papers (Simple Uniform Hashing Assumption) \cite{dietzfelbinger1990new,pagh2007linear,pagh2008uniform,dietzfelbinger2009applications}.}
  Concerning $G₂$, imagine that $M$ is revealed first. Conditioned on $M$, $G₂$ is composed of $n$ uniformly random edges like $G₁$.
  Concerning $G₃$, the key observation is that $M$ is a uniformly random matching even if the number to be randomly matched in every round is chosen by an adversary. A formal proof could consider any arbitrary adversarial strategy and use induction.
  \hfill
\end{proof}

\begin{lemma}
  \label{lem:numberOfOrientations}
  Let $G_n$ be the random graph sampled from the configuration model with $n$ nodes of degree $2$, i.e.\ the $2$-regular graph obtained by randomly joining $2n$ stubs that are evenly distributed among $n$ nodes. Then the number $c(G_n)$ of components of $G$ satisfies
  $𝔼(2^{c(G_n)}) ≤ e·\sqrt{2n}$.
\end{lemma}
We remark that a similar proof shows that $\mathds{E}(c(G_n)) \in \Oh(\log n)$.
Note also the similarity to the locker puzzle, which analyzes the length of the largest cycle in a random permutation \cite{stanley2011enumerative}.
\begin{proof}
  We will find a recurrence for $d_n := 𝔼(2^{c(G_n)})$.
  Consider an arbitrary node $v$ of $G_n$ and one of the stubs at $v$.
  This stub forms an edge with some other stub.
  We have $n-1$ other nodes, each with $2$ stubs, and we have the second stub at $v$. Each of these $2n-1$ stubs is matched with $v$ with equal probability. Therefore, the probability that $v$ has a self-loop is $\frac{1}{2n-1}$.

  (1) Conditioned on $v$ having a self-loop, we have found an isolated node. The distribution of the remaining graph is that of $G_{n-1}$ and the conditional expectation of $2^{c(G)}$ is therefore $𝔼(2^{1+c(G_{n-1})}) = 2d_{n-1}$.
  
  (2) Now condition on the formed edge connecting $v$ to $w ≠ v$. 
  We can now merge the nodes to a single one without affecting the number of components. The merged node inherits two unused stubs, one from $v$ and one from $w$.
  The distribution of the remaining graph is that of $G_{n-1}$, so in this case, the conditional expectation of $2^{c(G)}$ is simply $d_{n-1}$.

  These two cases lead us to the following recurrence:
  \begin{align*}
    d_{n}=\tfrac{1}{2n-1}2d_{n-1}+\big(1-\tfrac{1}{2n-1}\big)d_{n-1}=\big(1+\tfrac{1}{2n-1}\big)d_{n-1}.
  \end{align*}
  With the base case $d_0=1$, we can solve the recurrence and bound its value as follows, using that $\ln(1+x) \leq x$ for $x \geq 0$ as well as $H_n := \sum_{i = 1}^n \frac 1i ≤ 1 + \ln n$:
  \begin{align*}
    d_n
    &= \prod_{i=1}^n \left( 1 + \frac{1}{2i-1} \right)
    = \textrm{exp}\left(\sum_{i=1}^n \ln\left(1 + \frac{1}{2i-1}\right)\right) \\
    &\leq \textrm{exp}\left(\sum_{i=1}^n \frac{1}{2i-1}\right)
    \leq \textrm{exp}\left((1+H_{2n-1})/2\right)\\
    &\leq \textrm{exp}\left(1 + \ln(2n)/2\right)
    \leq e·\sqrt{2n}.
  \end{align*}
  \vspace{-2.7\baselineskip}\\\phantom{foo}\hfill
\end{proof}

Let us now re-state the known result proven in
RecSplit \cite{esposito2020recsplit}, which bounds
the probability that a random function is minimal
perfect.  We will use this later when proving
\cref{lem:constructionTries}.

\begin{lemma}[see \cite{esposito2020recsplit}]
  \label{lem:bruteForce}
  A random function $h: S \rightarrow [n]$ on a set $S$ of $n$ keys, is minimal perfect
  (i.e.\ is a bijection)
  with probability $e^{-n}\sqrt{2 \pi n}\cdot (1+o(1))$.
\end{lemma}
\begin{proof}
  Given $S$, there are
  $n^n$ possible functions from $S$ to~$[n]$
  and $n!$ of them are bijective.
  Therefore, the probability that a randomly selected function is minimal perfect is $n!/n^n$.
  The claim is obtained by applying Stirling's approximation.
  \hfill
\end{proof}

Now let us continue with the main theoretical contribution of this paper.

\begin{theorem}
  \label{lem:constructionTries}
  Let $h₀,h₁ : S → [n]$ be uniformly random functions. The probability that there exists $f : S → \{0,1\}$ such that $x ↦ h_{f(x)}(x)$ is bijective is at least $(e/2)^{-n} e^{-1} \sqrt{\pi}$.
\end{theorem}
\def\perf{\mathrm{perf}}
\begin{proof}
  Recall our shorthand $\orient(f)$ for the event that $x ↦ h_{f(x)}(x)$ is bijective.

  For any non-negative random variable $X$ we have
  \[
    𝔼(X) = ℙ(X > 0)·𝔼(X \mid X > 0).
  \]
  Setting $X = \#\{f \mid \orient(f)\}$ and rearranging this for $ℙ(X > 0)$ yields
  \[
    ℙ(∃f: \orient(f)) = \frac{𝔼(\#\{f \mid \orient(f)\})}{𝔼(\#\{f \mid \orient(f)\} \mid ∃f: \orient(f)).}
  \]
  We will consider the numerator and denominator in turn.

  \myparagraph{Numerator: Expectation.}
  Linearity of expectation (holding even for dependent variables) yields
  \begin{gather*}
    𝔼(\#\{f \mid \orient(f)\}) = \sum_{f} ℙ(\orient(f))\\
    = \sum_{f} ℙ(x ↦ h_{f(x)}(x) \text { is bijective on $S$}).
  \end{gather*}
  For any fixed $f$, the function $x ↦ h_{f(x)}(x)$ assigns independent random numbers to each $x ∈ S$, i.e.\ is a random function as considered in \cref{lem:bruteForce} and hence bijective with probability $e^{-n}\sqrt{2 \pi n}·(1+o(1))$. We therefore get
  \begin{align}
    \label{eq:numOfOrientations}
    \mathds{E}\left(\# \{f \mid \orient(f)\} \right) ≥ 2^n \cdot e^{-n}\sqrt{2 \pi n}.
  \end{align}

  \myparagraph{Denominator: Conditional Expectation.}
  Using observations (\ref{eq:pseudoforests-orientable}) and (\ref{eq:orientations-cycles}) we can shift our attention onto the graph $G$:
  \[
    𝔼(\#\{f \mid \orient(f)\} \mid ∃f: \orient(f))
    = 𝔼(2^{c(G)} \mid \PF(G)).
  \]
  By virtue of \cref{lem:probability-spaces} we can moreover move to a configuration model à la $G₃$. We first reveal the locations $x₁,…,x_{2n}$ of the $2n$ stubs (hence the degree sequence of $G₃$) and then consider the following \emph{peeling process} \cite{walzer2021peeling,Molloy05:Cores-in-random-hypergraphs,Luczak:A-simple-solution} that reveals edges of $G₃$ and simplifies $G₃$ in a step-by-step fashion.

  As long as there exists a node $v$ with only one stub, firstly, match it to a random stub to form a corresponding edge $\{v,w\}$ (consuming the two stubs) and, secondly, remove the node $v$ and the newly formed edge $\{v,w\}$. These removals do not affect the number of components of the resulting graph (since $v$ was connected to $w$), nor whether the resulting graph is a pseudoforest (since the component of $w$ lost one node and one edge).

  Let $n'$ be the number of nodes that remain after peeling and let $G'$ be the graph obtained by matching the remaining stubs. As discussed we have $\PF(G₃) ⇔ \PF(G')$ and $c(G') = c(G₃)$. Since the average degree of $G₃$ is $2$ and since we removed one node and one edge in every round, the average degree of $G'$ is also $2$. There are two cases.
  \vspace{-1mm}
  \begin{description}\itemsep-1.5mm\topsep-2mm 
    \item[Case 1:] Some node of $G'$ has degree $0$. Then $¬\PF(G')$ because some component of $G'$ must have average degree $> 2$.
    \item[Case 2:] No node of $G'$ has degree $0$. Since we ran the peeling process, there is also no node of $G'$ with degree $1$. Hence, every node of $G'$ has degree $2$. This makes $G'$ a collection of cycles. In particular $\PF(G')$ holds. Moreover, the generation of $G'$ is precisely the situation discussed in \cref{lem:numberOfOrientations}.
  \end{description}
  \vspace{-1mm}
  Because the two cases imply opposite results on $G'$ being a pseudoforest, we know that $\PF(G')$ holds if \emph{and only if} we arrive in Case 2.
  While we have no understanding of the distribution of $n'$, we can nevertheless compute:
  \begin{align}
    \label{eq:conditionalPseudotree}
    𝔼(2^{c(G)} &\mid \PF(G)) = 𝔼(2^{c(G₃)} \mid \PF(G₃))\notag\\
    &= 𝔼(2^{c(G')} \mid \PF(G')) = 𝔼(2^{c(G')} \mid \text{Case 2})\notag\\
    &≤ \max_{1 ≤ i ≤ n} 𝔼(2^{c(G')} \mid \text{Case 2 with $n' = i$})\notag\\
    &≤ \max_{1 ≤ i ≤ n} e \sqrt{2i} = e \sqrt{2n}.
  \end{align}
  \myparagraph{Putting the Observations Together.}
  Combining our bounds on numerator \ref{eq:numOfOrientations} and denominator \ref{eq:conditionalPseudotree} gives the final result
  \begin{align*}
    ℙ(∃f: \orient(f)) &≥ 2^n e^{-n}\sqrt{2 \pi n} / (e \sqrt{2n})\\
                      &= (e/2)^{-n}e^{-1}\sqrt{\pi}.
  \end{align*}
  \vspace{-2.5\baselineskip}\\\phantom{foo}\hfill
\end{proof}

ShockHash tries different hash function seeds,
which is equivalent to generating random graphs.
Given the probability that the random graph is a
pseudoforest, it is easy to determine the expected
number of graphs ShockHash needs to try in order
to find an MPHF.  This leads directly to the space
usage and construction time of ShockHash, which we analyze in
the following.

\begin{theorem}
  \label{thm:shockHashConstruction}
  A ShockHash minimal perfect hash function mapping $n$ keys to $[n]$ needs $\log_2(e)n + o(n)$ bits of space in expectation and can be constructed in expected time $\Oh((e/2)^n \cdot n)$.
\end{theorem}
\begin{proof}
  From \cref{lem:constructionTries}, we know that
  the probability of the graph being 1-orientable
  is $\geq (e/2)^{-n} e^{-1}\sqrt{\pi}$.
  We construct these graphs uniformly at random,
  so the expected number of seeds to try is $\leq
  (e/2)^{n} e / \sqrt{\pi}$.  The space usage is
  given by the $n + o(n)$ bits for the retrieval
  data structure, plus the bits to store the hash
  function~index:
  \begin{align*}
    &\mathds{E}(\log_2(\textrm{seeds to try}))
    \overset{*}{\leq}
    \log_2(\mathds{E}(\textrm{seeds to
      try}))\\ &\leq \log_2\left((e/2)^{n}
    e / \sqrt{\pi}\right) = \log_2(e)n - n + \Oh(1).
  \end{align*}
  In the step annotated with $*$, we use Jensen's
  inequality \cite{jensen1906fonctions} and the
  fact that $\log_2$ is concave.

  For determining if any of the $2^n$ functions corresponding to such a seed is valid, we can use an algorithm for finding connected components, as described in \cref{s:sccFilter}.
  This takes linear time for each of the seeds, resulting in an overall construction time of $\Oh((e/2)^n \cdot n)$.
  Constructing the retrieval data structure is then possible in linear time \cite{dillinger2022burr} and happens only once, so it is irrelevant for the asymptotic time here.
  \hfill
\end{proof}

For the brute-force approach, each of the $e^n/\sqrt{2 \pi n}$ expected trials needs $n$ hash function evaluations, leading to a construction time of $\Oh(e^n \sqrt{n})$.
Now, as shown in \cref{thm:shockHashConstruction}, ShockHash needs time $\Oh((e/2)^n \cdot n)$.
This makes ShockHash almost $2^n$ times faster than the previous state of the art.
Given the observations in Ref. \cite{bez2022high}, we conjecture that ShockHash with rotation fitting reduces the number of hash function evaluations by an additional factor of $n$, while the space overhead tends to zero.

\section{ShockHash-RS = ShockHash + RecSplit}\label{s:shockhashInRecSplit}
As mentioned in the introduction, real world MPHF
constructions usually do not search for a function
for the entire input set directly.  Instead, they
partition the input of size $N$ and then search
on smaller subproblems of size $n$.  Even though ShockHash
demonstrates significant speedups, by itself, it
still needs exponential running~time.

To demonstrate the usefulness of ShockHash in
practice, we integrate it as a base case into the
highly space efficient RecSplit framework (see
\cref{s:recsplit}) and obtain ShockHash-RS.  We keep the general structure
of RecSplit intact, and only replace the bijection
search in the leaves.  For each leaf, we store the
mapping from its keys to their hash function
indices.  Finally, after all leaves are processed,
we construct the 1-bit retrieval data structure
with all the $N$ entries together.

\myparagraph{Fanouts.}  RecSplit tries to balance
the difficulty between the splittings and the
bijections.  ShockHash improves the performance of
the bijections significantly but does not modify
the way that the splittings are calculated.  In
this paper, we focus only on the
bijections. Similar techniques may work for splitting
with \cref{s:splittings} outlining first ideas.

To balance the amount of work done between
splittings and bijections, we need to adapt the
splitting parameters using the same techniques
as the RecSplit paper \cite{esposito2020recsplit}.
Instead of fanouts $\lceil 0.35 n + 0.5
\rceil$ and $\lceil 0.21 n + 0.9 \rceil$ for
the two last splitting levels (see
\cref{s:recsplit}), our numerical evaluation gives
$\lfloor 0.10
n + 0.5\rfloor$ and $\lfloor 0.073 n + 0.9
\rfloor$ for ShockHash-RS.  However, preliminary experiments show
that this is not optimal in practice.  ShockHash,
especially with rotation fitting, is so much
faster that the additional time invested into the
splittings does not pay off.  We find
experimentally that setting the lowest splitting
level to 4 and the second lowest to 3
achieves much better results in practice.  To also
provide faster and space-inefficient
configurations, we set all fanouts to~2 when selecting leaf size $n \leq 24$.

\myparagraph{SIMD Parallelization.}
\label{s:simd}
In ShockHash, we use SIMD parallelism in two locations.
First, we use SIMD to hash all keys and determine the bit mask of the two sets.
A key point here is to collect the bitwise \texttt{OR} of individual lanes and to only add the lanes together after all keys are done.
Second, we use SIMD to evaluate the bit mask filter (see \cref{s:bitmaskFilter}) with different rotations in parallel.
Our implementation uses AVX-512 (8 64-bit values) if available and AVX2 (4 64-bit values) otherwise.

\section{Refinements, Variants and Future Directions}\label{s:refinements}
In the following section, we describe variants and implementation details of ShockHash.
Additionally, we give ideas for future directions.

\myparagraph{Implementation Details.}
Determining whether a given graph is a pseudoforest can be achieved in linear time using a connected components algorithm.
However, this is likely not practical because the graph data structure needs to be built.
In practice, we therefore use incremental cuckoo hash table construction with near linear time.
In \cref{s:unionFindFilter}, we describe another approach using a union-find data structure that is particularly interesting for the rotation fitting variant.
However, in our experiments, we find incremental cuckoo hash table construction to be more efficient.
In \cref{s:partialHashCalculation}, we describe a way to reduce the number of hash function evaluations in the rotation fitting variant.
The idea is to fix the hash function for the first group of keys over a segment of $k$ tried seeds.
For large $n$, this saves considerable time while having negligible impact on the space consumption.

\myparagraph{Parallelization.}
While our implementation of ShockHash-RS supports basic multi-threaded construction based on SIMDRecSplit \cite{bez2022high}, we do not discuss this here.
A missing piece for a full parallelization is the BuRR retrieval data structure \cite{dillinger2022burr} which is currently only available as sequential code.
BuRR construction is parallelizable on $p$ cores with additional space overhead of about $p$ machine words.
However, for highly space-efficient MPHFs, which are the main focus of this paper, constructing the retrieval data structure is not a bottleneck.
GPU parallelization might be difficult and inefficient for cuckoo hashing as it has irregular control flow and memory access.
Since filtering asymptotically dominates the computations for highly space-efficient variants, one might look at a hybrid implementation where a GPU produces a stream of seeds defining random graphs that cover all nodes and where a multicore CPU performs further stages of computation.

\myparagraph{Fast Splitting for RecSplit.}\label{ss:fastSplitting}
ShockHash-RS
significantly improves bijection search within the
RecSplit framework which can make brute-force search for
splittings a significant part of the running time
in some configurations. In \cref{s:splittings} we therefore outline
how to find a splitting with a constant number of trials in expectation.
The idea is to replace the binary splitting hash function by one with a larger range and to store a compressed threshold value that leads to an exact split.

\myparagraph{Beyond RecSplit.}
ShockHash-RS suffers from comparatively large query times as RecSplit queries have to traverse a splitting tree, decoding variable-bitlength data on each level. This overhead is not inherent in ShockHash itself and we could also integrate it into other frameworks.
One way is to just use \emph{two levels of hash-based splitting} allowing for variable sized leaves (see \cref{ss:FastShockHash}).
ShockHash can also be adapted to the \emph{Hash-and-Displace} framework (see \cref{s:related}).
For each bucket, we now look for hash function pairs
whose overall result is a pseudoforest. Once more, the remaining
decisions are made using cuckoo hashing and stored in a retrieval data structure.
Compared to splitting based approaches, this is perhaps more elegant as the global hash range alleviates the need for storing and using splitting information like prefix sums.
Refer to \cref{ss:ShockDisplace} for more details.

\section{Experiments}\label{s:experiments}
We run our experiments on an Intel i7 11700 processor with 8 cores and a base clock speed of 2.5 GHz.
The machine runs Ubuntu 22.04 with Linux 5.15.0 and supports AVX-512 instructions.
We use the GNU C++ compiler version 11.2.0 with optimization flags \texttt{-O3 -march=native}.
For the competitors written in Rust, we compile in release mode with \texttt{target-cpu=native}.
As input data, we use strings of uniform random length $\in [10, 50]$ containing random characters except for the zero byte.
Note that, as a first step, almost all compared codes generate a \emph{master hash code} of each key using a high quality hash function.
Any possible additional hash function can then be evaluated on the master hash code in constant time, independent of the input distribution.
All experiments use a single thread.
While almost all compared codes have a multi-threaded implementation, and perfect hashing can be parallelized trivially by partitioning, this is not the focus here.
The code and scripts needed to reproduce our experiments are available on GitHub under the General Public License \cite{sourceCode,sourceCodeCompetitors}.

\begin{figure}[t]
  \centering
      \begin{tikzpicture}
        \begin{axis}[
            xlabel={$n$},
            ylabel={Avg. successful seed},
            plotHfEvals,
            ymode=log,
            ymax=1e10,
            legend columns=4,
            transpose legend,
            legend to name=legendHashEvals,
          ]
          \addplot+[color=colorBruteForce] coordinates { (6,64.8549) (8,416.091) (10,2767.63) (12,18417.6) (14,127722.0) (16,856746.0) (18,6.48792e+06) (20,4.09977e+07) (22,2.05908e+08) };
          \addlegendentry{Brute-force \cite{esposito2020recsplit}};
          \addplot+[color=colorRotationFitting] coordinates { (6,13.3014) (8,56.3617) (10,285.461) (12,1573.24) (14,8911.15) (16,54628.5) (18,369627.0) (20,1.9798e+06) (22,1.99377e+07) (24,7.45176e+07) };
          \addlegendentry{Rotation fitting \cite{bez2022high}};
          \addplot+[mark=triangle,color=black] coordinates { (6,2.2757) (8,3.8287) (10,6.73595) (12,11.5547) (14,20.3844) (16,36.5247) (18,64.9892) (20,116.499) (22,208.758) (24,374.636) (26,680.764) (28,1221.88) (30,2238.19) (32,4023.07) (34,7519.32) (36,13354.2) (38,24472.3) (40,44311.0) };
          \addlegendentry{ShockHash};
          \addplot+[mark=shockhash,color=colorShockHash] coordinates { (6,1.12525) (8,1.23685) (10,1.4658) (12,1.85955) (14,2.40425) (16,3.3074) (18,4.78105) (20,7.2033) (22,10.88855) (24,17.4039) (26,27.6061) (28,45.1996) (30,75.3199) (32,127.268) (34,223.582) (36,357.877) (38,625.271) (40,1065.71) };
          \addlegendentry{ShockHash + RF};

          \addplot[mark=none,densely dotted,color=colorBruteForce] coordinates { (5,26.4787) (7,165.357) (10,2778.78) (15,336730) (20,4.32797e+07) (30,7.78366e+11) (40,1.48477e+16) };
          \addlegendentry{$e^n/\sqrt{2 \pi n}$};
          \addplot[mark=none,densely dotted,color=colorRotationFitting] coordinates { (5,5.29575) (7,23.6224) (10,277.878) (15,22448.7) (20,2.16399e+06) (30,2.59455e+10) (40,3.71193e+14) };
          \addlegendentry{$e^n/\sqrt{2 \pi n}/n$};
          \addplot[mark=none,densely dotted,color=black] coordinates { (5,7.11282) (7,13.1393) (10,32.9886) (15,152.998) (20,709.593) (30,15263.5) (40,328321) };
          \addlegendentry{$(e/2)^ne/\sqrt{\pi}$};
          \addplot[mark=none,densely dotted,color=colorShockHash] coordinates { (5,1.42256) (7,1.87704) (10,3.29886) (15,10.1999) (20,35.4796) (30,508.783) (40,8208.03) };
          \addlegendentry{$(e/2)^ne/(n\sqrt{\pi})$};
        \end{axis}
    \end{tikzpicture}
    \hspace{-3mm}
    \begin{tikzpicture}
        \begin{axis}[
            plotHfEvals,
            title={},
            xlabel={$n$},
            ylabel={Bits space overhead}, %
            cycle list name=myColorList,
          ]
          \addplot+[color=colorBruteForce] coordinates { (6,-0.0234508) (8,-0.0100292) (10,0.00662729) (12,-0.00246884) (14,0.00832896) (16,-0.0123481) (18,-0.0181847) (20,0.0455501) (22,0.0731805) };
          \addlegendentry{Brute-Force  \cite{esposito2020recsplit}};
          \addplot+[mark=square,color=colorRotationFitting] coordinates { (6,0.236825) (8,0.103994) (10,0.0384379) (12,0.0303382) (14,-0.0208516) (16,0.00583817) (18,0.0344676) (20,-0.0344609) (22,0.0682468) };
          \addlegendentry{Rotation Fitting \cite{bez2022high}};
          \addplot+[mark=triangle,color=black] coordinates { (6,0.351081) (8,0.818735) (10,1.06764) (12,1.22911) (14,1.32902) (16,1.39139) (18,1.45315) (20,1.49068) (22,1.52505) (24,1.57209) (26,1.57358) (28,1.61109) (30,1.63094) (32,1.69517) (34,1.65146) (36,1.68107) (38,1.68893) (40,1.71764) };
          \addlegendentry{ShockHash};
          \addplot+[mark=shockhash,color=colorShockHash] coordinates { (6,1.06005) (8,1.47376) (10,1.67663) (12,1.76716) (14,1.7863) (16,1.78234) (18,1.76139) (20,1.73096) (22,1.72061) (24,1.68402) (26,1.6575) (28,1.675) (30,1.64282) (32,1.63045) (34,1.67673) (36,1.67297) (38,1.68878) (40,1.6969) };
          \addlegendentry{ShockHash + RF};

          \legend{};
        \end{axis}
    \end{tikzpicture}

    \begin{tikzpicture}
        \ref*{legendHashEvals}
    \end{tikzpicture}
  \caption{Left: Average successful seed of ShockHash compared with more simple brute-force techniques. Right: Idealized space overhead over the lower bound $log_2(n^n/n!)$ in bits. If the average seed is $s$ we charge $\log_2(s)$ bits, plus $n$ bits for retrieval (if applicable).}
  \label{fig:hashFunctionEvals}
\end{figure}
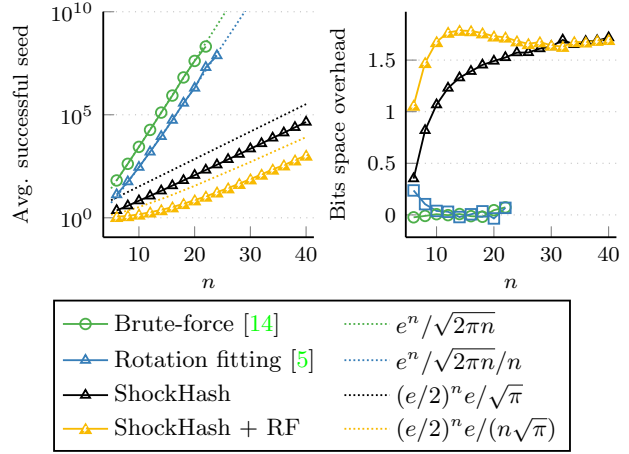

\subsection{Number of Trials in Theory and Practice.}
In \cref{fig:hashFunctionEvals}, we compare the
average number of hash function trials over
multiple runs for each bijection search technique.
From the different slopes of the curves, it is
clearly visible that rotation fitting
\cite{bez2022high} saves a polynomial factor compared to plain brute-force,
while ShockHash saves an exponential factor.
Additionally, we plot the shown upper bounds for the
number of trials of brute-force and ShockHash.
For the rotation fitting
variants, we plot the base variants divided by
$n$, which is not formally shown to be a
theoretical bound, but is an obvious conjecture.
The plot shows that brute-force and rotation
fitting are close to the given functions.  For
ShockHash, the measurements are even better than
the theory, which indicates that our proof in
\cref{lem:constructionTries} is not tight.
Surprisingly, ShockHash seems to match the
function we get when dividing our analysis by $\sqrt{n}$.
We conjecture
that the expected number of 1-orientations of a random
pseudoforest might actually not be $e\cdot\sqrt{2n}$, but close to constant.  This
makes ShockHash an even better replacement for the
brute-force technique.

\Cref{fig:hashFunctionEvals} also gives the difference between the idealized space consumption and the space lower bound $\log_2(n^n/n!)$.
It indicates that ShockHash loses space sublinear in $n$, which becomes negligible for larger $n$.
However, this explains why we need to select larger $n$ in ShockHash-RS compared to RecSplit to achieve the same space consumption per key.
Even with these larger $n$, ShockHash construction is significantly faster than brute-force.

\subsection{ShockHash-RS.}
In this section, we evaluate ShockHash-RS, which uses ShockHash as a base case in the RecSplit framework.
For BuRR retrieval \cite{dillinger2022burr} we use 2-bit bumping info and 128-bit words (64 bits for $n\leq 24$).
To partition keys to buckets, we sort them by their key using IPS$^2$Ra \cite{axtmann2020engineering}.
We also apply the idea to reduce hash function evaluations by fixing the hash function for half of the keys over multiple search iterations.
Refer to \cref{s:partialHashCalculation} for more details.

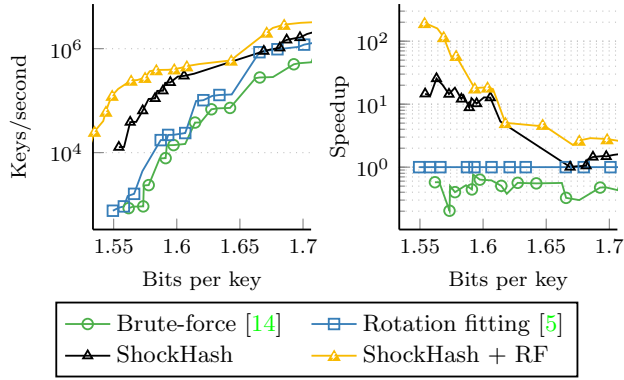
\begin{figure}[t]
  \centering
      \centering
    \begin{tikzpicture}
        \begin{axis}[
            xlabel={Bits per key},
            ylabel={Keys/second},
            plotLeafMethods,
            xmin=1.54,
            xmax=1.7,
            ymode=log,
            legend to name=paretoLeafMethodsLegend,
            legend columns=2,
            mark repeat*=2,
          ]

          \addplot+[color=colorBruteForce] coordinates { (1.56187,855.165) (1.56532,913.505) (1.57341,924.95) (1.57387,2338.81) (1.57774,2372.54) (1.58647,7742.93) (1.59115,7837.42) (1.59211,13875.8) (1.59744,13966.3) (1.60665,14693.3) (1.61444,37916.1) (1.61881,38028.6) (1.62742,68240.8) (1.63226,71068.2) (1.64217,72505.8) (1.6624,271518) (1.66575,279252) (1.6762,284657) (1.69271,501253) (1.69725,526039) (1.70668,547945) (1.7108,1.05708e+06) (1.71625,1.1655e+06) (1.72474,1.26582e+06) (1.77763,1.53139e+06) (1.78249,1.75747e+06) (1.79287,1.97239e+06) (1.85913,2.53165e+06) (1.87058,2.97619e+06) (1.88029,3.59712e+06) (1.94514,4.40529e+06) (1.94641,5.07614e+06) (2.01165,6.71141e+06) (2.21301,8.69565e+06) };
          \addlegendentry{Brute-force \cite{esposito2020recsplit}};
          \addplot+[mark=square,color=colorRotationFitting] coordinates { (1.54937,762.965) (1.5539,869.904) (1.55807,924.288) (1.56111,1464.52) (1.56618,1612.49) (1.57274,4392.71) (1.58726,17416.1) (1.59199,17764.8) (1.59316,21926) (1.59771,22053.1) (1.60669,23839.6) (1.6153,98668) (1.62106,102124) (1.6281,123457) (1.63376,128535) (1.64351,131614) (1.66547,855432) (1.67102,912409) (1.67973,975610) (1.69652,1.08342e+06) (1.7009,1.1976e+06) (1.70938,1.3089e+06) (1.72552,1.45138e+06) (1.73125,1.63666e+06) (1.74175,1.81159e+06) (1.79046,2.05339e+06) (1.7971,2.44499e+06) (1.80611,2.88184e+06) (1.87396,4.20168e+06) (1.94279,4.56621e+06) (2.00047,4.80769e+06) (2.01128,7.04225e+06) (2.06834,7.75194e+06) (2.31379,8.92857e+06) };
          \addlegendentry{Rotation fitting \cite{bez2022high}};
          \addplot+[mark=triangle,color=black] coordinates { (1.55409,12734.5) (1.55822,12877.3) (1.56309,38665.3) (1.56787,39117.5) (1.57298,64180.7) (1.5777,107875) (1.58283,110448) (1.58452,154226) (1.58907,155039) (1.59079,220264) (1.59489,224568) (1.60069,289855) (1.60575,299760) (1.61369,329924) (1.66896,894454) (1.67362,958773) (1.68165,1.0395e+06) (1.68283,1.31062e+06) (1.68704,1.4771e+06) (1.69526,1.61551e+06) (1.69758,1.63934e+06) (1.70316,1.91571e+06) (1.71077,2.15983e+06) (1.72186,2.28311e+06) (1.72971,2.69542e+06) (1.75478,3.06748e+06) (1.78713,3.37838e+06) (1.79001,3.47222e+06) (1.81581,4.23729e+06) (1.84874,4.6729e+06) (1.89422,5.05051e+06) (1.95481,5.29101e+06) (2.05775,5.52486e+06) (2.37803,5.52486e+06) };
          \addlegendentry{ShockHash};
          \addplot+[mark=shockhash,color=colorShockHash] coordinates { (1.52457,3068.67) (1.52672,5192.7) (1.53,8966.04) (1.53216,15598.2) (1.53474,25455.7) (1.54172,40472.7) (1.54381,61888.8) (1.54633,91382.6) (1.54874,125313) (1.55358,168407) (1.56376,244738) (1.56871,249252) (1.57478,277085) (1.57855,368732) (1.58372,389408) (1.59332,392773) (1.59792,403714) (1.60346,428816) (1.60787,461894) (1.61763,505817) (1.64305,595593) (1.64711,723589) (1.67111,2.05761e+06) (1.67614,2.51256e+06) (1.68498,2.90698e+06) (1.69819,3.14465e+06) (1.71955,3.33333e+06) (1.74229,3.663e+06) (1.75757,4.13223e+06) (1.77769,4.5045e+06) (1.80326,4.83092e+06) (1.83799,5.34759e+06) (2.11603,5.52486e+06) };
          \addlegendentry{ShockHash + RF};

        \end{axis}
    \end{tikzpicture}
    \hspace{-3mm}
    \begin{tikzpicture}
        \begin{axis}[
            xlabel={Bits per key},
            ylabel={Speedup},
            plotLeafMethods,
            xmin=1.54,
            xmax=1.7,
            ymode=log,
            mark repeat*=2,
            ylabel style={yshift=-1ex},
          ]

          \addplot+[color=colorBruteForce] coordinates { (1.56187,0.575890) (1.56532,0.576227) (1.57341,0.203299) (1.57387,0.501446) (1.57774,0.401031) (1.58647,0.516298) (1.59115,0.442745) (1.59211,0.765878) (1.59744,0.633519) (1.60665,0.616564) (1.61444,0.504759) (1.61881,0.377471) (1.62742,0.563903) (1.63226,0.558937) (1.64217,0.552713) (1.6624,0.561437) (1.66575,0.325417) (1.6762,0.299964) (1.69271,0.474259) (1.69725,0.477819) (1.70668,0.431017) (1.7108,0.800636) (1.71625,0.853236) (1.72474,0.876740) (1.77763,0.772007) (1.78249,0.874579) (1.79287,0.904714) (1.85913,0.662846) (1.87058,0.724494) (1.88029,0.849830) (1.94514,0.962783) (1.94641,1.108171) (2.01165,0.952454) (2.21301,1.034609) };
          \addlegendentry{bruteforce};
          \addplot+[mark=square,color=colorRotationFitting] coordinates { (1.54937,1.000000) (1.5539,1.000000) (1.55807,1.000000) (1.56111,1.000000) (1.56618,1.000000) (1.57274,1.000000) (1.58726,1.000000) (1.59199,1.000000) (1.59316,1.000000) (1.59771,1.000000) (1.60669,1.000000) (1.6153,1.000000) (1.62106,1.000000) (1.6281,1.000000) (1.63376,1.000000) (1.64351,1.000000) (1.66547,1.000000) (1.67102,1.000000) (1.67973,1.000000) (1.69652,1.000000) (1.7009,1.000000) (1.70938,1.000000) (1.72552,1.000000) (1.73125,1.000000) (1.74175,1.000000) (1.79046,1.000000) (1.7971,1.000000) (1.80611,1.000000) (1.87396,1.000000) (1.94279,1.000000) (2.00047,1.000000) (2.01128,1.000000) (2.06834,1.000000) (2.31379,1.000000) };
          \addlegendentry{rotations};
          \addplot+[mark=triangle,color=black] coordinates { (1.55409,14.5997) (1.55822,13.6786) (1.56309,25.4552) (1.56787,20.3035) (1.57298,14.4302) (1.5777,18.2847) (1.58283,12.0781) (1.58452,13.8096) (1.58907,8.835153) (1.59079,12.4619) (1.59489,10.2196) (1.60069,12.8166) (1.60575,12.6807) (1.61369,5.306363) (1.66896,1.004558) (1.67362,1.030495) (1.68165,1.053364) (1.68283,1.318699) (1.68704,1.448436) (1.69526,1.503481) (1.69758,1.478202) (1.70316,1.563367) (1.71077,1.636157) (1.72186,1.611889) (1.72971,1.703405) (1.75478,1.639915) (1.78713,1.660283) (1.79001,1.693057) (1.81581,1.404310) (1.84874,1.301475) (1.89422,1.173775) (1.95481,1.146602) (2.05775,0.726037) (2.37803,0.000000) };
          \addlegendentry{plain};
          \addplot+[mark=shockhash,color=colorShockHash] coordinates { (1.52457,0.000000) (1.52672,0.000000) (1.53,0.000000) (1.53216,0.000000) (1.53474,0.000000) (1.54172,0.000000) (1.54381,0.000000) (1.54633,0.000000) (1.54874,0.000000) (1.55358,195.509) (1.56376,159.096) (1.56871,116.844) (1.57478,56.4514) (1.57855,58.8251) (1.58372,38.5206) (1.59332,17.91) (1.59792,18.2743) (1.60346,18.5117) (1.60787,17.3613) (1.61763,5.056273) (1.64305,4.530427) (1.64711,4.735214) (1.67111,2.253634) (1.67614,2.648905) (1.68498,2.886936) (1.69819,2.797005) (1.71955,2.389139) (1.74229,2.019338) (1.75757,2.193758) (1.77769,2.270453) (1.80326,1.771070) (1.83799,1.581739) (2.11603,0.694458) };
          \addlegendentry{plainRotate};

          \legend{};
        \end{axis}
    \end{tikzpicture}

    \begin{tikzpicture}
        \ref*{paretoLeafMethodsLegend}
    \end{tikzpicture}
  \caption{Space versus construction time of RecSplit, with different base case methods of calculating the leaves plugged in. Basic version without SIMD parallelization, $N=1$ million. The plot on the right gives speedups\footref{fn:paretoSpeedups} relative to the current state of the art, the rotation fitting method \cite{bez2022high}.
   We plot all Pareto optimal data points but only show markers for every second point to increase readability. Therefore, the lines might bend on positions without markers.}
  \label{fig:leafMethods}
\end{figure}

\addtocounter{footnote}{1}
\footnotetext{\label{fn:paretoSpeedups}Note that
  giving speedups is non-trivial here because
  there might not be a configuration that achieves
  the same space usage that we could compare
  with. We therefore calculate the speedup
  relative to an interpolation of the next larger
  and next smaller data points. This is reasonable
  since RecSplit instances can be interpolated as
  well by hashing a certain fraction of keys into
  data structures with different configurations.}

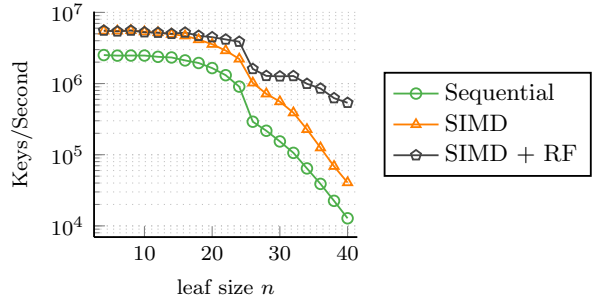
\begin{figure}[t]
  \centering
      \centering
    \hspace{8mm}

    \begin{tikzpicture}
        \begin{axis}[
            xlabel={leaf size $n$},
            plotParameters,
            ymode=log,
            ylabel={Keys/Second},
            legend style={at={(1.1,0.5)},anchor=west},
            legend columns=1,
          ]
          \addplot coordinates { (4,2.51256e+06) (6,2.46914e+06) (8,2.47525e+06) (10,2.47525e+06) (12,2.3753e+06) (14,2.32558e+06) (16,2.11416e+06) (18,1.94553e+06) (20,1.65017e+06) (22,1.31062e+06) (24,905797) (26,290951) (28,217155) (30,153704) (32,106191) (34,64189) (36,38922.6) (38,22486.1) (40,12786.9) };
          \addlegendentry{Sequential};
          \addplot coordinates { (4,5.49451e+06) (6,5.49451e+06) (8,5.58659e+06) (10,5.34759e+06) (12,5.20833e+06) (14,4.92611e+06) (16,4.69484e+06) (18,4.14938e+06) (20,3.57143e+06) (22,2.90698e+06) (24,2.21729e+06) (26,1.02249e+06) (28,717360) (30,557414) (32,389257) (34,226860) (36,125360) (38,68785.3) (40,40586.1) };
          \addlegendentry{SIMD};
          \addplot coordinates { (4,5.55556e+06) (6,5.40541e+06) (8,5.55556e+06) (10,5.26316e+06) (12,5.18135e+06) (14,5e+06) (16,5.20833e+06) (18,4.69484e+06) (20,4.4843e+06) (22,4.16667e+06) (24,3.87597e+06) (26,1.60514e+06) (28,1.28205e+06) (30,1.26263e+06) (32,1.27714e+06) (34,993049) (36,851789) (38,624220) (40,536769) };
          \addlegendentry{SIMD + RF};
        \end{axis}
    \end{tikzpicture}
  \caption{Construction performance of ShockHash-RS using different leaf sizes. With and without SIMD parallelization. $N=1$ million keys, $b=2000$.}
  \label{fig:parameters}
\end{figure}

\Cref{fig:leafMethods} compares the space versus construction time tradeoff
of plain brute-force
\cite{esposito2020recsplit}, brute-force with rotation fitting
\cite{bez2022high}, and ShockHash (with and without rotation fitting).
For each method we measure many parameter settings and only plot those that are
Pareto-optimal, i.e., not dominated by another configuration of the same method regarding both space and construction time. 
 Using the
same amount of space, basic ShockHash can
construct MPHFs up to \maxSpeedupNonSimdPlain{}
times faster than brute-force with rotation fitting.  ShockHash
with rotation fitting improves that to a
factor of up to \maxSpeedupNonSimdPlainRotate{}.
Here, ShockHash further profits from the
fact that filtering the rotations can happen
efficiently in registers.

\Cref{fig:parameters} plots the construction performance of ShockHash-RS with different leaf sizes with and without SIMD parallelization, as well as with rotation fitting.
ShockHash-RS mainly profits from rotation fitting if the leaf size is $n>20$ -- for small leaves it is likely that several rotations cover all leaf positions.
Indeed, the solution of $(1-1/e^2)^n = 1/n$ is about 21 (see \cref{lem:filterProbability}).
\Cref{fig:parameters} also shows that the SIMD parallelization improves the construction time significantly for all values of $n$.
The drop of construction performance at $n \geq 25$ is caused by the fact that we switch the splitting strategy (and the BuRR configuration) for increased space efficiency.

\subsection{Comparison with Competitors.}
We now compare ShockHash-RS with competitors from the
literature.  Competitors include CHD
\cite{belazzougui2009hash}, SicHash
\cite{lehmann2023sichash}, PTHash
\cite{pibiri2021pthash}, FMPHGO
\cite{beling2023fingerprinting}, RecSplit
\cite{esposito2020recsplit}, and SIMDRecSplit
\cite{bez2022high}.  We do not plot BBHash
\cite{limasset2017fast} because it is
significantly outperformed by
FMPH \cite{beling2023fingerprinting},
another implementation of the same technique.  While
SIMDRecSplit also includes a fast GPU
implementation, we do not compare it here, as it
would be unfair because of the different hardware
architecture.

\begin{figure}[t]
  \centering \centering
\def\diff{1.46}
    \begin{tikzpicture}
        \begin{axis}[
            plotPareto,
            xlabel={Bits/key},
            ylabel={Throughput (Keys/s)},
            legend to name=paretoLegend,
            legend columns=2,
            xmin=1.50-\diff,
            xmode=log,
            ymode=log,
            minor xtick={1.48-\diff,1.49-\diff,1.51-\diff,1.52-\diff,1.53-\diff,1.54-\diff,1.55-\diff,1.56-\diff,1.57-\diff,1.58-\diff,1.59-\diff},
            xtick=      {1.5-\diff, 1.6-\diff, 1.7-\diff, 1.8-\diff, 1.9-\diff, 2.0-\diff, 2.1-\diff, 2.2-\diff, 2.3-\diff, 2.4-\diff, 2.5-\diff, 2.6-\diff, 2.7-\diff,
                         2.8-\diff, 2.9-\diff, 3-\diff, 3.1-\diff, 3.2-\diff, 3.3-\diff, 3.4-\diff, 3.5-\diff},
            xticklabels={1.5,       1.6,       1.7,       1.8,       ~,         2.0,       ~,         ~,         ~,         ~,         2.5,       ~,         ~,
                         ~,         ~,         3,       ~,         ~,         ~,         ~,         ~},
          ]
          \addplot[mark=x,color=colorChd,solid] coordinates { (0.54275,162161) (0.58152,288484) (0.58288,569930) (0.62322,863036) (0.64531,1.60953e+06) (0.68796,2.07943e+06) (0.73418,3.35683e+06) (0.76632,3.8373e+06) (0.84023,5.09424e+06) (0.87893,5.40833e+06) (0.99038,5.85138e+06) (1.12562,6.08273e+06) (1.17309,6.23053e+06) (1.30193,6.42674e+06) };
          \addlegendentry{CHD \cite{belazzougui2009hash}};
          \addplot[mark=asterisk,color=colorRustFmph,solid] coordinates { (1.34355,1.14025e+07) (1.36692,1.26582e+07) (1.42765,1.41844e+07) (1.51519,1.53846e+07) (1.61441,1.65563e+07) (1.727,1.75439e+07) (1.84778,1.85185e+07) };
          \addlegendentry{FMPH \cite{beling2023fingerprinting}};
          \addplot[mark=diamond,color=colorRustFmphGo,solid] coordinates { (0.75308,6.94927e+06) (0.78116,7.93651e+06) (0.84524,8.69565e+06) (0.93367,9.32836e+06) (1.04083,9.82318e+06) (1.16255,1.01215e+07) (1.29584,1.03734e+07) };
          \addlegendentry{FMPHGO \cite{beling2023fingerprinting}};
          \addplot[mark=pentagon,color=colorPthash,solid,mark repeat*=4] coordinates { (0.54215,326755) (0.5817,437178) (0.5925,706564) (0.62976,903914) (0.64301,1.241e+06) (0.67741,1.51883e+06) (0.69528,1.86289e+06) (0.72749,2.18914e+06) (0.73565,2.50125e+06) (0.79117,2.8393e+06) (0.79553,3.10752e+06) (0.8245,3.44116e+06) (0.83696,3.65764e+06) (0.88978,3.98089e+06) (0.90626,4.12882e+06) (0.92272,4.42478e+06) (0.93673,4.55996e+06) (0.97391,4.82859e+06) (0.98848,4.91642e+06) (1.03968,5.31067e+06) (1.06628,5.47046e+06) (1.08645,5.48246e+06) (1.10863,5.70451e+06) (1.12776,5.70776e+06) (1.16196,5.90667e+06) (1.17215,5.95238e+06) (1.21449,6.14628e+06) (1.24206,6.33312e+06) (1.28059,6.48088e+06) (1.32528,6.58328e+06) (1.3739,6.72495e+06) (1.4286,6.83527e+06) (1.45656,6.86342e+06) (1.48246,6.97837e+06) (1.51179,7.02741e+06) (1.54681,7.09723e+06) (1.57104,7.24113e+06) (1.93032,7.26216e+06) };
          \addlegendentry{PTHash \cite{pibiri2021pthash}};
          \addplot[mark=square,color=colorRecSplit,solid] coordinates { (0.1244,7902.7) (0.13014,14056.1) (0.14551,14894.7) (0.15281,37885.4) (0.16548,68818.9) (0.18095,72320.2) (0.24938,1.0227e+06) (0.26455,1.20963e+06) (0.31662,1.43699e+06) (0.33048,1.85908e+06) (0.39987,2.3084e+06) (0.41797,3.24254e+06) (0.48451,3.8373e+06) (0.48891,4.2123e+06) (0.55193,5.29101e+06) (0.62979,5.75705e+06) (0.75408,6.38162e+06) (0.82698,6.95894e+06) (1.46786,6.96379e+06) (1.78648,7.44048e+06) };
          \addlegendentry{RecSplit \cite{esposito2020recsplit}};
          \addplot[mark=+,color=colorSimdRecSplit,solid] coordinates { (0.10005,7278.14) (0.11168,20867.6) (0.12519,85321.3) (0.13051,114911) (0.14591,125518) (0.15396,568150) (0.16679,633955) (0.18239,690369) (0.20372,3.99202e+06) (0.21852,4.44444e+06) (0.23476,5.00751e+06) (0.24902,5.7241e+06) (0.26484,6.36943e+06) (0.28015,7.43494e+06) (0.33013,7.5358e+06) (0.34415,9.01713e+06) (0.41357,1.08342e+07) (0.48294,1.16959e+07) (0.55197,1.28041e+07) (0.61215,1.2987e+07) (0.63028,1.36054e+07) (0.6915,1.38696e+07) (0.99747,1.50376e+07) (1.11351,1.63399e+07) (1.38263,1.70358e+07) (1.50042,1.71233e+07) };
          \addlegendentry{SIMDRecSplit \cite{bez2022high}};
          \addplot[mark=shockhash,color=colorShockHash,solid,mark repeat*=4] coordinates { (0.06306,8961.08) (0.06443,15933.7) (0.06644,26258.2) (0.06919,42659.9) (0.07195,70807.5) (0.07843,121448) (0.08095,177528) (0.08389,262715) (0.08725,358474) (0.09135,530560) (0.10191,838715) (0.11212,985902) (0.11614,1.26279e+06) (0.12967,1.35044e+06) (0.14629,1.42268e+06) (0.18757,1.72652e+06) (0.20127,1.74703e+06) (0.2102,3.7594e+06) (0.22285,4.09333e+06) (0.28307,4.3122e+06) (0.33052,4.84966e+06) (0.3445,5.40249e+06) (0.41101,5.78369e+06) (0.48526,5.97015e+06) (0.56741,6.05694e+06) (0.65573,6.43501e+06) };
          \addlegendentry{ShockHash-RS};
          \addplot[mark=o,color=colorSicHash,solid,mark repeat*=4] coordinates { (0.55421,4.52694e+06) (0.58393,5.79039e+06) (0.61508,6.28931e+06) (0.62381,6.43501e+06) (0.65432,6.52316e+06) (0.65436,6.75219e+06) (0.6723,6.77048e+06) (0.68,7.00771e+06) (0.68443,7.18907e+06) (0.7107,7.53012e+06) (0.74119,7.69823e+06) (0.74206,8.1103e+06) (0.77204,8.19001e+06) (0.77273,8.56164e+06) (0.80348,8.62813e+06) (0.82892,8.67303e+06) (0.83421,8.73362e+06) (0.86031,8.81057e+06) (0.86436,8.88099e+06) (0.89025,8.99281e+06) (0.92143,9.06618e+06) (0.92604,9.07441e+06) (0.95194,9.27644e+06) (0.95679,9.44287e+06) (1.01352,9.46074e+06) (1.01759,9.61538e+06) (1.0748,9.72763e+06) (1.16679,9.94036e+06) (1.22785,1.00604e+07) (1.2589,1.01317e+07) (1.28893,1.0142e+07) (1.35069,1.01626e+07) (1.44241,1.02145e+07) (1.93303,1.02564e+07) };
          \addlegendentry{SicHash \cite{lehmann2023sichash}};
        \end{axis}
    \end{tikzpicture}

    \begin{tikzpicture}
        \ref*{paretoLegend}
    \end{tikzpicture}
  \caption{Pareto front of the space usage of different competitors.
    $N=10$ million keys.
    Note that both axes are logarithmic.
    For ShockHash-RS, we use ShockHash with SIMD and rotation fitting inside RecSplit.
    For ShockHash-RS, SicHash and PTHash, we
    plot all Pareto optimal data points but only
    show markers for every fourth point to
    increase readability. Therefore, the lines
    might bend on positions without markers.}
  \label{fig:pareto}
\end{figure}
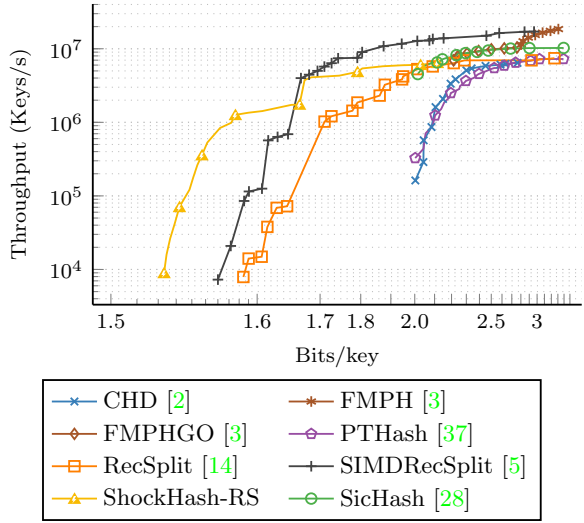

\myparagraph{Construction.}
\Cref{fig:pareto} compares the construction throughput of different competitors.
Note that we are mainly interested in the space usage configurations close to the theoretical lower bound.
Therefore, the Figure uses a logarithmic x-axis.
Only the RecSplit based competitors achieve space usage below 1.9 bits per key.
At configurations with less
than 1.65 bits per key, ShockHash-RS can
significantly outperform SIMDRecSplit.
ShockHash-RS is focused even more on the space efficient configurations than SIMDRecSplit.
Therefore, it does not achieve the
same throughput as SIMDRecSplit for the less space
efficient configurations.  This is not surprising
because with these configurations, searching for
bijections is fast, so constructing the retrieval
data structure has a significant performance
penalty.  If SIMD operations are not available for ShockHash-RS and competitors,
ShockHash-RS shows even more significant speedups to
the next best competitor, as already discussed in
the previous section.

\Cref{tab:queries} gives a selection of typical configurations.
For SicHash \cite{lehmann2023sichash}, PTHash \cite{pibiri2021pthash} and FMPHGO \cite{beling2023fingerprinting}, we use the configurations given in the original papers, where FMPHGO is configured to use the hash cache.
For RecSplit-based techniques, we mainly use space-efficient configurations with $b=2000$ and select the leaf size $n$ such that a similar space consumption is achieved.
Comparing the configurations with
a space consumption of 1.56 bit per key, ShockHash-RS
is about \speedupSIMD156{} times faster than the
next competitor.  Comparing different methods where each is given about 10 minutes of construction time,
RecSplit is able to produce a perfect hash
function with 1.58 bits per key.  During the
course of three years, this space consumption
was improved to 1.56 bits per key (SIMDRecSplit)
and now to only 1.52 bits per key (ShockHash-RS).  Compared to
SIMDRecSplit, ShockHash-RS makes the gap to the lower
space bound of $\approx 1.442$ bits per key about
30\% smaller.  Using a single CPU thread,
ShockHash-RS achieves a space usage close to what was
previously only achieved using a GPU
\cite{bez2022high}.

\myparagraph{Queries.}
\Cref{tab:queries} also shows the query throughput of typical configurations.
RecSplit-based techniques have slower queries than the other techniques as they have to traverse several levels of a tree, decoding variable-bitlength data in each step.
ShockHash-RS additionally needs to access a retrieval data structure.
However, when comparing configurations that
achieve a similar space efficiency, the query
performance of ShockHash-RS is similar to
competitors.
This shows that the overhead of the
retrieval operation is
small compared to the work for traversing the
heavily compressed tree.
When query performance is the main concern, PTHash \cite{pibiri2021pthash} trades space usage for much faster queries.

\begin{table}[t]
  \caption{Query and construction performance of
    typical configurations of ShockHash-RS and
    competitors. $N=10$ million keys.
    Space is given in bits per key and construction time is given in ns per key.
    For ShockHash-RS, we use ShockHash with SIMD and rotation fitting inside RecSplit.
    }
  \label{tab:queries}
  \centering 
\addtolength\tabcolsep{-3.6pt}
\small
\begin{tabular}[t]{ll rr}
    \toprule
    Method & Space & Constr. & Query \\ \midrule

               FMPHGO, $\gamma$=$2.0, s$=$4, b$=$16$ & 2.860 &       90 ns &  53 ns \\
               FMPHGO, $\gamma$=$1.0, s$=$4, b$=$16$ & 2.212 &      135 ns &  69 ns \\\midrule
            PTHash, $c$=$11.0$, $\alpha$=$0.88$, D-D & 4.379 &      135 ns &  25 ns \\
             PTHash, $c$=$7.0$, $\alpha$=$0.99$, C-C & 3.524 &      198 ns &  20 ns \\
              PTHash, $c$=$6.0$, $\alpha$=$0.99$, EF & 2.345 &      247 ns &  34 ns \\\midrule
     SicHash, $\alpha$=$0.9$, $p_1$=$21$, $p_2$=$78$ & 2.412 &      116 ns &  40 ns \\
    SicHash, $\alpha$=$0.97$, $p_1$=$45$, $p_2$=$31$ & 2.082 &      169 ns &  41 ns \\\midrule
                        RecSplit, $n$=$8$, $b$=$100$ & 1.792 &      713 ns &  74 ns \\
                      RecSplit, $n$=$14$, $b$=$2000$ & 1.585 & 125\,521 ns &  97 ns \\\midrule
                    SIMDRecSplit, $n$=$8$, $b$=$100$ & 1.808 &      117 ns &  80 ns \\
                  SIMDRecSplit, $n$=$14$, $b$=$2000$ & 1.585 &  11\,749 ns & 108 ns \\
                  SIMDRecSplit, $n$=$16$, $b$=$2000$ & 1.560 & 137\,902 ns & 100 ns \\\midrule
          \textbf{ShockHash-RS}, $n$=$30$, $b$=$100$ & 1.654 &      563 ns &  82 ns \\
         \textbf{ShockHash-RS}, $n$=$30$, $b$=$2000$ & 1.583 &      787 ns & 114 ns \\
         \textbf{ShockHash-RS}, $n$=$39$, $b$=$2000$ & 1.556 &   1\,805 ns & 118 ns \\\midrule
         \textbf{ShockHash-RS}, $n$=$58$, $b$=$2000$ & 1.523 & 111\,593 ns & 115 ns \\

    \bottomrule
\end{tabular}

\end{table}

\section{Conclusion and Future Work}\label{s:conclusion}
By combining trial-and-error search with retrieval
data structures computed using cuckoo hashing,
ShockHash achieves an exponential speedup over plain
brute-force (almost a factor $2^n$).
This enables
the currently most work-efficient way
to achieve near space-optimal minimal perfect
hash functions and breaks the dominance of the
previous best methods that relied on pure brute-force for their base-case subproblems.

ShockHash-RS (i.e.\ ShockHash in the RecSplit framework) is up to two
orders of magnitude faster than the state of the art when comparing
sequential codes (with and without SIMD
acceleration).  We expect that this will extend to
parallel multicore and GPU implementations, at
least when looking for highly space-efficient
functions.

It would be interesting to use ShockHash outside
the RecSplit framework in order to accelerate
query times (likely tolerating slightly higher
space consumption).  This could be done directly
by using a faster and more ``flat'' decomposition of
the input problem into subproblems. ShockHash
makes this more promising since it can support
larger base cases than brute-force. One could also
generalize perfect hash functions based on
Hash-and-Displace
\cite{belazzougui2009hash,pibiri2021pthash} to
search for pseudo-forests rather than perfect hash
functions directly.

We believe that further exponential reduction of the search space based on ShockHash will be possible in the future.
A first step in that direction is \emph{bipartite} ShockHash, which we developed during the review period.
Instead of sampling random graphs, bipartite ShockHash samples \emph{bipartite} random graphs -- it uses two hash functions of range $[n/2]$, where one of them is shifted by $n/2$.
During search, it builds a pool of hash function candidates and then tries all pairs that can be formed between them.
By filtering the candidates \emph{before} combining them, this achieves an additional exponential speedup on top of the one achieved by ShockHash.
Further work is needed to properly tune and analyze bipartite ShockHash, but we give an initial description in a technical report \cite{lehmann2023bipartite}.

\myparagraph{Acknowledgements.}
This project has received funding from the European Research Council (ERC) under the European Union’s Horizon 2020 research and innovation programme (grant agreement No. 882500) as well as from the German Research Foundation (DFG) grant 465963632.

\begin{center}
  \includegraphics[width=4cm]{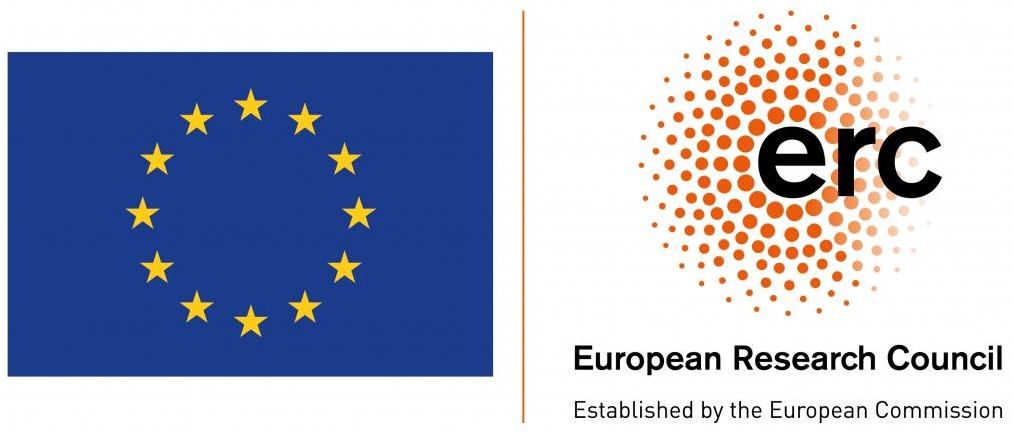}
\end{center}

\bibliographystyle{plainurl}
\bibliography{paper}

\clearpage

\appendix
\renewcommand{\thefigure}{A.\arabic{figure}}
\setcounter{figure}{0}
\renewcommand{\thetable}{A.\arabic{table}}
\setcounter{table}{0}

\noindent{\Large\bf Appendix}

\section{More Details on Refinements}\label{s:refinementsAppendix}
In this section, we give additional details on the refinements and ideas outlined in \cref{s:refinements}.

\subsection{Construction using Union-Find.}\label{s:unionFindFilter}
An alternative to using connected components or cuckoo hashing for construction is to use a union-find data structure.
A union-find data structure manages disjoint subsets of a given set, where initially each keys is its own subset.
The $\textrm{union}(x, y)$ operation merges the two subsets containing $x$ and $y$.
The $\textrm{find}(x)$ operation returns a representative for the set containing $x$.
For this data structure, there is a very simple implementation based on building trees during $\textrm{union}$ operations and collapsing them during $\textrm{find}$ operations.
With this, $m$ $\textrm{union}$ and $\textrm{find}$ operations take time $\Oh(m\cdot\alpha(m))$, where $\alpha$ is the inverse Ackermann function \cite{tarjan1975efficiency}.
The inverse Ackermann function grows extremely slowly, and has a value of less than $5$ for any reasonable input size.

During ShockHash construction, if we look at the edges individually, an edge that connects one tree and one (pseudo)tree results in a larger (pseudo)tree.
Connecting two nodes in a tree makes that tree a pseudotree.
Finally, an edge that connects two pseudotrees creates a structure that is no longer a pseudotree.
Using a union-find data structure, we can start with $n$ individual sets, where each representative is labeled as a tree.
We iteratively add additional edges and update the labels of the representatives to indicate whether a node is a tree or a pseudotree.
Once we try to union two pseudotrees, we know that the given seed cannot lead to successful construction.

With rotation fitting, this becomes significantly more interesting.
Because both sets remain constant, we can determine the union-find data structure for the larger set.
The union-find data structure can then be re-used between different rotations.
For the construction, it does not matter which set we rotate, because we can always apply the reverse rotation to the other set.
Inserting the smaller set for each rotation value can determine an invalid seed more quickly.
In practice, however, this does not improve the performance because the bit mask filters out most seeds anyway and does not require the more expensive union-find operations.

\subsection{Partial Hash Calculation.}\label{s:partialHashCalculation}
Hashing the first set of keys during rotation fitting
almost always yields a graph that, by itself, is a pseudoforest.
The reason is that this is usually close to the load threshold $c=0.5$ and $n$ is small (which enables higher load \cite{lehmann2023sichash}).
Therefore, we can keep the hashes for the first set the same and just retry hash functions for the second set.
More precisely, if $x$ is the hash function seed, we hash each key in the first set with seed $x-(x\textrm{ mod }k)$, where $k$ is a tuning parameter, and the keys of the second set with seed $x$.
Therefore, the hash values of the first set can be cached over multiple iterations.
In preliminary experiments, we find a value of $k=8$ to be a good fit -- values much larger than that have diminishing returns in performance improvement and start to influence the space consumption.
At $k=8$, however, the influence on the space consumption is negligible when $n$ is large.
Given that hashing the keys is a bottleneck during construction, this reduces the number of keys that need to be hashed by a factor of close to $2$.
We only apply this optimization for large $n > 32$.

\subsection{Fast Splitting for RecSplit.}\label{s:splittings}
Besides brute-force construction of leaves, the
near-optimal space-efficiency of RecSplit
\cite{esposito2020recsplit} also hinges on
brute-force splitting which is optimal up to a
constant number of bits per split.  Concretely, a
binary split takes $\Oh(\sqrt{n})$ trials of 1-bit
hash functions (i.e.\ work $n^{3/2}$) and needs
$\log_2(n)/2+\Oh(1)$ bits to store the seed in
expectation.

We now outline how to achieve space
$\log_2(n)/2+\Oh(1)$ for a binary split with expected
running time $\Oh(n)$.
Assume $n$ is even to avoid trivial rounding issues.
The idea is to try splitting
hash functions $h_s$ with range $[cn]$ for some
constant $c$ where $cn$ is even. Now we look for a threshold value
$t$ such that $h_s(x)\leq t$ for exactly $n/2$ keys of $S$.  Such a threshold
exists with constant probability. If it exists, it
can be found in linear time using an appropriate
selection algorithm. Storing a seed leading to a successful split takes $\Oh(1)$ bits
and storing the threshold would need $\log_2 cn=\log_2 n + \Oh(1)$ bits.

We can reduce this to $\log_2(n)/2+\Oh(1)$ by
observing that for any constant $d$,
the threshold will lie in a range of size $d\sqrt{n}$ around $cn/2$
with constant probability
(the number of keys below $cn/2$ obeys a
binomial distribution with expectation $n/2$).
Thresholds in this range can be represented by just
storing the difference to $cn/2$ which takes
$\log_2(n/2)+\Oh(1)$ bits.  If the threshold is
outside this range, the seed failed and a new seed
is tried.  Overall, we have constant success
probability (exact split by threshold possible and
the threshold value is within the assumed range)
so that we still need $\Oh(1)$ bits for the seed and
$\log_2(n/2)+\Oh(1)$ bits overall.
Note that by appropriately defining the assumed threshold range,
we can use fixed-width binary encoding of the threshold so that
the query cost for the split will be similar to the brute-force case.

We leave a more detailed analysis for future work
that could compare the space overheads of threshold-based and brute-force splitting.
Note that for $c=1/n$ this method reverts
to the brute-force approach, i.e., it is to be
expected that threshold based splitting can closely approach the space
consumption for brute-force splitting.

\subsection{Faster Queries without RecSplit.}\label{ss:FastShockHash}

We can look for a different tradeoff between space
and query performance by avoiding the recursive
splitting used in RecSplit. We can replace it by
1-level splitting of each bucket, i.e., overall a
2-level hierarchy (that is also frequently used in
compressed data structures like rank-select).  In
the most simple case, one would simply split each
bucket randomly into leaves and store prefix sums
of bucket sizes ($\log_2 N$ bits) as well as local prefix
sums of leaf sizes (about $\log_2 b$ bits).
Storing these prefix sums takes about
$$\frac{N}{b}\log_2 N+\frac{N}{n}\log_2 b$$
bits. We can ensure that this is small compared to
the lower space bound of $N\log_2 e$, when $b\gg
\log N$ and $n\gg \log\log N$.

Implementing this simple approach directly is
likely to suffer from large construction time for
the largest leaves. There are various ways to
limit maximum leaf sizes. For example, we could
store thresholds generalizing the approach to
binary splitting in \cref{s:splittings}.  Or we
could bump some keys from buckets or leaves
similar to the approach used in BuRR and its
variants \cite{dillinger2022burr}.  We could also
store per-bucket seeds to select a bucket local
hash function that does not produce oversize
leaves. This approach could be aided by choosing
the number of leaves in a bucket proportional to
the bucket size.  Such measures cause additional space
overhead, but it does not change the basic analysis.
Also, correlations between the balancing
information and the prefix sum values open many
opportunities for compression. Overall, we get a
quite large design space whose exploration we
leave to future work.

\subsection{Pseudoforest Hash-and-Displace.}\label{ss:ShockDisplace}

We could combine the idea of ShockHash with the
idea of Hash and Displace outlined in
\cref{s:related}: Keys are hashed to buckets
for which we encode a seed specifying \emph{two}
hash functions. Construction searches for seeds
that result in a pseudoforest. This should have a
success probability that is much larger than
directly finding an injective function --
resulting in good space efficiency combined with
fast construction. Query time will lie in between
ShockHash-RS and basic
Hash-and-Displace.

\section{Filter strength}\label{s:FilterStrength}
In \cref{s:bitmaskFilter}, we describe a filter based on bit masks that enables to skip cuckoo hash table construction for most of the hash function seeds.
We then show that the probability for a seed to pass the filter is at most $≈0.864^n$ (see \cref{lem:filterProbability}).
In the following, we now give the exact solution that also takes into account the correlation between the bits.

\begin{lemma}
  Assume $2n$ balls are randomly thrown into $n$ bins. The probability that all bins receive at least one ball is $Θ(b^n)$ where
  $b = 2e^λ/(λe²)$ and where $λ$ is the unique solution to $2 = {λ}/(1-e^{-λ})$.
  Numerical approximation gives $λ ≈ 1.597$ and $b ≈ 0.836$.
\end{lemma}

\begin{proof}
  \def\x{\vec{x}}
  \def\X{\vec{X}}
  \def\Y{\vec{Y}}
  \def\Z{\vec{Z}}
  \def\Po{\mathrm{Po}}
  \def\Var{\mathrm{Var}}
  Consider the following random variables:
  \begin{itemize}
    • $X_i$ for $i ∈ [n]$ is the number of balls in bin $i$. The sequence $(X₁,…,Xₙ)$ follows a multinomial distribution.
    • $Y₁,…,Yₙ \sim \Po(2)$ are i.i.d.\ Poisson random variables.
    • $Z₁,…,Zₙ \sim D(λ)$ are i.i.d.\ random variables with distribution
      \[\Pr_{Z \sim D(λ)}[Z = i] =
        \begin{cases}
          0 & \text{ if $i = 0$,}\\
          \frac{e^{-λ}λ^i}{(1-e^{-λ})·i!} & \text{ if $i ∈ ℕ$}
        \end{cases}
      \]
      The idea is that $Z \sim D(λ)$ is like a Poisson random variable (which attains value $i$ with probability $e^{-λ}λ^i/i!$) but conditioned to be positive. We have $𝔼[Z] = λ/(1-e^{-λ})$ and choose $λ$ such that $𝔼[Z] = 2$.
    • $N_Y = \sum_{i = 1}^n Y_i$, $N_Z = \sum_{i = 1}^n Z_i$.
  \end{itemize}
  We use vector notation $\vec{·}$ as a shorthand for corresponding sequences of numbers, e.g. $\X = (X₁,…,Xₙ)$. We denote by $R$ the set of all outcomes of the balls-into-bins experiment and by $R₊$ the outcomes with at least one ball per bin, i.e.
  \begin{align*}
    R &= \{ \x ∈ ℕ₀^n \mid \sum_{i = 1}^n x_i = 2n\},\\
    R₊ &= \{ \x ∈ ℕ^n \mid \sum_{i = 1}^n x_i = 2n\}.
  \end{align*}
  Finally, the following abbreviations will be useful 
  \begin{align*}
    q(\x) = \prod_{i = 1}^n \frac{1}{x_i!},
    \ \ q(R) = \sum_{\x ∈ R} q(\x),
    \ \ q(R₊) = \sum_{\x ∈ R₊} q(\x).
  \end{align*}
  The random variables $\X$, $\Y$ and $\Z$ are closely related: The distributions of $\Y$ and $\X$ coincide when conditioning on $\{N_Y = 2n\}$ and the distributions of $\Y$ and $\Z$ coincide when conditioning on $\{Y ∈ R⁺\}$ and $\{N_Z = 2n\}$. Our argument will rest on understanding normalisation terms mediating between the three settings, which are numbers $C₁,C₂,C₃$ (that may depend on $n$ but not on $\x$) such that for all $\x ∈ R$ we have
  \begin{align}
    \Pr[\X = \x] &= C₁·q(\x)\label{eq:C1}\\
    \Pr[\Y = \x \mid N_Y = 2n] &= C₂·q(\x)\label{eq:C2}\\
    \Pr[\Z = \x \mid N_Z = 2n] &= C₃·q(\x)·𝟙_{\{\x ∈ R₊\}}\label{eq:C3}
  \end{align}
  Let us verify this claim and compute $C₁$, $C₂$ and $C₃$.
  \begin{align*}
    \Pr[\X &= \x] = \binom{2n}{x₁\ x₂\ …\ x_n} n^{2n}\\
    &= \frac{(2n)!}{x₁!·x₂!·…·xₙ!} n^{-2n} = q(\x) · \underbrace{\frac{(2n)!}{n^{2n}}}_{= C₁}\\
    \Pr[\Y &= \x \mid N_Y = 2n] = \frac{\Pr[\Y = \x]}{\Pr[N_Y = 2n]}\\
       &= \frac{\prod_{i = 1}^n \Pr[Y_i = x_i]}{\Pr[N_Y = 2n]}
       = \frac{\prod_{i = 1}^n e^{-2}·\frac{2^{x_i}}{x_i!}}{\Pr[N_Y = 2n]}\\
       &= \frac{e^{-2n}·2^{2n}\prod_{i = 1}^n \frac{1}{x_i!}}{\Pr[N_Y = 2n]}
       = q(\x)\underbrace{\frac{2^{2n}}{e^{2n}\Pr[N_Y = 2n]}}_{= C₂}
  \end{align*}
  In \cref{eq:C3} we get $0$ on both sides if $\x ∈ R \setminus R₊$. Consider now $\x ∈ R₊$.
  \begin{align*}
    \Pr[\Z &= \x \mid N_Z = 2n] = \frac{\Pr[\Z = \x]}{\Pr[N_Z = 2n]}\\
       &= \frac{\prod_{i = 1}^n \Pr[Z_i = x_i]}{\Pr[N_Z = 2n]}
       = \frac{\prod_{i = 1}^n \frac{e^{-λ}}{1-e^{-λ}}·\frac{λ^{x_i}}{x_i!}}{\Pr[N_Z = 2n]}\\
       &= q(\x) \underbrace{\frac{λ^{2n}}{e^{λn}(1-e^{-λ})^n \Pr[N_Z = 2n]}}_{= C₃}
  \end{align*}
  By summing \cref{eq:C1} and \cref{eq:C2} over all $\x ∈ R$ we obtain
  \begin{gather}
    1 = C₁·q(R) \text{ and } 1 = C₂·q(R)\notag\\
    \text{ and hence } C₁ = C₂.\label{eq:C1C2}
  \end{gather}
  By summing \cref{eq:C1} and \cref{eq:C3} over all $\x ∈ R₊$ we obtain
  \begin{gather}
    \Pr[\X ∈ R₊] = C₁·q(R₊) \text{ and } 1 = C₃·q(R₊)\notag\\
    \text{ and hence } \Pr[\X ∈ R₊] = C₁/C₃.\label{eq:C1C3}
  \end{gather}
  Putting \cref{eq:C1C2,eq:C1C3} together gives:
  \begin{align*}
    \Pr[\X &∈ R₊] = C₁/C₃ = C₂/C₃\\
    &= \frac{2^{2n}}{e^{2n}\Pr[N_Y = 2n]}\frac{e^{λn}(1-e^{-λ})^n \Pr[N_Z = 2n]}{λ^{2n}}\\
    & \Big(\frac{2^{2}e^{λ}(1-e^{-λ})}{e^{2}λ^{2}}\Big)^n \frac{\Pr[N_Z = 2n]}{\Pr[N_Y = 2n]}\\
    & \Big(\frac{2e^{λ}}{e^{2}λ}\Big)^n \frac{\Pr[N_Z = 2n]}{\Pr[N_Y = 2n]}
    = b^n \frac{\Pr[N_Z = 2n]}{\Pr[N_Y = 2n]}.
  \end{align*}
  The last step is to show that the two probabilities are of the same magnitude. Take $N_Y$ first. It is the sum of $n$ independent random variables with constant variance. The central limit theorem suggests that the histogram of $N_Y$ has a bell-curve-shape with mean $2n$ and standard deviation $Θ(\sqrt{n})$, though a formal proof has to exploit that the greatest common divisor of the support of the underlying distribution $\Po(2)$ is $1$.
  This implies that $\Pr[N_Y = 2n] = Θ(1/\sqrt{n})$. The same argument applies to $N_Z$ and the claim follows.
  \hfill
\end{proof}  

\end{document}